\documentclass[journal, comsoc]{IEEEtran}
\usepackage[T1]{fontenc}
 \IEEEoverridecommandlockouts

\usepackage{epsfig,rotating,setspace,latexsym,amsmath,epsf,amssymb,bm}
\usepackage{cite,graphicx,authblk,color, subcaption}

\usepackage{cleveref}
\usepackage{float}
\usepackage{url}

\usepackage{mdwmath}
\usepackage{mdwtab}
\usepackage{mathtools}

\usepackage[font=footnotesize,skip=3pt]{caption}

\newtheorem{manualtheoreminner}{Example}
\newenvironment{manualexample}[1]{%
  \renewcommand\themanualtheoreminner{#1}%
  \manualtheoreminner
}{\endmanualtheoreminner}

\DeclarePairedDelimiterX\Basics[1](){ #1}

\usepackage{balance}
\usepackage{enumerate}

\allowdisplaybreaks

\usepackage{tabularx}
\usepackage[makeroom]{cancel}
\usepackage{multirow}

\usepackage[linesnumbered,ruled,vlined]{algorithm2e}

  \newtheorem{myexp}{Example}

\newtheorem{theorem}{Theorem}
\newtheorem{remark}{Remark}
\newtheorem{lemma}{Lemma}

\newtheorem{example}{Example}

\newtheorem{definition}{Definition}
\newenvironment{proof}[1]{\medskip\par\noindent
{\bf Proof:\,}\,#1}{{\mbox{\,$\blacksquare$}\par}}

\newcommand*{\QEDA}{\hfill\ensuremath{\blacksquare}}%
\usepackage[dvipsnames]{xcolor}

\usepackage{comment}

\begin{document}

\title{Topological Interference Management with Confidential Messages}

\author{
Jean de Dieu Mutangana \qquad Ravi Tandon\\ 
Department of Electrical and Computer Engineering\\
University of Arizona, Tucson, AZ, USA\\
E-mail: \{\textit{mutangana, tandonr}\}@email.arizona.edu
}
\maketitle
\begin{abstract}
\footnote{This work was supported by NSF grants CNS-1715947, CCF-2100013, and CAREER-1651492.}
The topological interference management (TIM) problem refers to the study of the $K$-user  partially connected  interference networks with no channel state information at the transmitters (CSIT), except for the knowledge of network topology. In this paper, we study the \textit{TIM problem with confidential messages} (TIM-CM), where message confidentiality must be satisfied in addition to reliability constraints. In particular,  each transmitted message must be decodable at its intended receiver and remain confidential at the remaining $(K-1)$ receivers.

Our main contribution is to present a comprehensive set of results for the TIM-CM problem by studying the symmetric secure degrees of freedom (SDoF). To this end, we first characterize necessary and sufficient conditions for feasibility of positive symmetric SDoF for any arbitrary topology. We next present two achievable schemes for the TIM-CM problem: For the first scheme, we use the concept of \textit{secure partition} and, for the second one, we use the concept of \textcolor{black}{\textit{secure independent sets}}.  We also present outer bounds on symmetric SDoF for any arbitrary network topology. Using these bounds, we characterize the optimal symmetric SDoF of all $K=2$-user and $K=3$-user network topologies.
\end{abstract}

\IEEEpeerreviewmaketitle

\section{Introduction}

\IEEEPARstart{P}{hysical} layer security (PHY-SEC) exploits the inherent randomness of the wireless channel such as fading or noise in order to  establish  secure communication between the legitimate network users. It was introduced by Wyner in his seminal work on the degraded wiretap channel \cite{Wyner}. Wyner's model was extended to the nondegraded wiretap channel in \cite{CsiszarKorner}, and then to the Gaussian wiretap channel in \cite{CheongHellman}.  Numerous extensions to other multi-terminal problems ensued, including secrecy constrained broadcast channels (BC), multiple access channels (MAC), multiple-input multiple-output (MIMO) channels, and multi-user interference channels (IC), e.g., references \cite{KhistiWornell, OggierHassibi, LiuShamai, MukherjeeTandonUlukus,  NafeaYener, NafeaYener1, LinLin, TYLiuUlukus, LiPetropulu2011, SLinPLin2013, ShlezingeDabora, MutanganaTandonTIFS, XieUlukus,  MutanganaTandonISIT, AttiaTandon, MutanganaTandonICC20}. For an overview of PHY-SEC research progress and potential applications in next generation wireless systems, we refer the reader to comprehensive surveys \cite{WuKhistiCaire, PoorSchaefer, YenerUlukus, AMukherjeeFakoorianSwinlehurst}.

A large body of works in PHY-SEC have been achieved under the assumption that channel state information at the transmitters (CSIT) is available --be it instantaneous \cite{KhistiWornell, OggierHassibi, LiuShamai}, delayed \cite{MaddahAliTse, TandonPoorShamai, LashgariAvestimehr, SeifTandonLi}, or alternating \cite{MukherjeeTandonUlukus, TandonJafarPoor}. Moreover, a major portion of prior works on multi-user interference networks (with or without secrecy constraints) have assumed that the network is \textit{fully connected}, i.e., that each transmitter in a given network is connected to every receiver in the same network and vice versa, e.g.,  \cite{MaddahAliTse, MutanganaTandonISIT, XieUlukus}. For instance, in \cite{XieUlukus} the authors studied the $K$-user fully connected interference channel with confidential messages (IC-CM) and showed that when CSIT is perfectly available, then the sum secure degrees of freedom (sum SDoF) scales linearly with  $K$. Additionally, it is known that, for the fully connected interference networks, SDoF is zero when CSIT is not available. This is due to the fact that, under such settings, all receivers are statistically equivalent from each transmitter's perspective which puts decodability and secrecy constraints in direct conflict. 

As indicated in \cite{WuKhistiCaire, PoorSchaefer, YenerUlukus, AMukherjeeFakoorianSwinlehurst}, a considerable portion of works on PHY-SEC have also studied secrecy constrained wireless channels in the absence of CSIT. For instance, in \cite{TYLiuUlukus} the authors characterized sum SDoF for the MIMO wiretap channel with no CSI anywhere for $T \geq 2\min (n_t ,n_r)$, where parameters $n_t$, $n_r$, and $T$ respectively represent the number of antennas at the transmitter, number of antennas at the receiver, and coherence time. In \cite{MutanganaTandonISIT} the authors showed that positive sum SDoF is achievable for the $K$-user fully connected IC-CM with no CSIT by leveraging the intersymbol interference (ISI) heterogeneity which is inherently present within the subclass of channels with memory.  Therein, it was shown that the achievable sum SDoF scales linearly with $K$ under some ISI heterogeneity conditions. Note that, in reality, the signals propagating over the wireless medium face several physical obstacles in addition to signal energy dissipation over traveled distance  \cite{Goldsmith, TseViswanathFundamentals}. Particularly, in multi-user interference networks, the above physical phenomena lead to the presence of \textit{weak channels}, i.e., channels whose signal strength is at noise floor levels or below a preset minimum power threshold  for reliable signal detection at the receivers. It is thus reasonable  to model multi-user wireless networks by discarding the weak channels, a setting referred to as that of \textit{partially connected} networks, i.e., where each transmitter is only connected to a subset of receivers and vice versa. 

Motivated by the above reasons, research considerations from the other extreme in terms of CSIT availability and network connectivity have recently gained traction under the framework of topological interference management (TIM) \cite{Jafar, MalekiJafarCadambe, Gesbert, NaderializadehAvestimehr, YiCaire, ShiZhangLetaief, LiuHanLiMa, AquilinaRatnarajah, DoumiatiAssaadArtail, YiSun}. In particular, the TIM problem studies partially connected multi-user interference networks with no CSIT, except for the knowledge of network topology and channel statistics. In other words, research on the TIM problem seeks to exploit the inherent heterogeneity due to the partial connectivity in order to ensure network reliability in the absence of CSIT.  We note here that the TIM problem is closely related to the index coding (ICOD) problem \cite{BirkandKol, Bar-YossefBirkJayramKol, HuangRouayheb, ArbabjolfaeiKim, ArbabjolfaeiHKimFund}, a framework under which the messages from all transmitters pass through a single common node (or equivalently, a server) which encodes them into a common function, that it then broadcasts to all receivers. Each receiver then applies its side information (also known as antidotes) to the broadcast function in order to decode its intended message. The authors in \cite{Jafar} and \cite{MalekiJafarCadambe} have shown that the TIM problem can be reformulated into the ICOD problem and vice versa. Moreover, due to the fact that the common node has access to all messages and the presence of antidotes at the receivers, the channel capacity for a given ICOD problem is an upper bound on the capacity of its TIM network counterpart. Note that the general ICOD and the general TIM are both still open problems. We refer the reader to \cite{ArbabjolfaeiHKimFund} and references therein for an overview of fundamentals of ICOD. A variant of the ICOD problem called pliable index coding (PICOD) was recently proposed \cite{BrahmaFragouli2015, BrahmaFragouli2018, Su2019, LiuTuninetti}. PICOD is a slightly more relaxed problem where each receiver is satisfied by decoding any one arbitrary message (not belonging to its side information set) from the broadcast function.

In this paper, we study the TIM problem with an additional secrecy constraint,  a framework that we refer to as \textit{TIM with confidential messages} (TIM-CM). Here, each transmitted message must be decodable at its intended receiver and remain confidential at the remaining $(K-1)$ receivers. The Besides a small number of recent results, e.g., \cite{AttiaTandon, MutanganaTandonICC20}, the TIM-CM problem has largely remained unexplored. The work in \cite{AttiaTandon} derived a lower and an upper bound on sum SDoF for the \textit{regular} TIM-CM problem. This is the secrecy constrained version of the well understood \textit{regular} TIM problem, where each user is assumed to receive signals from a constant number of transmitters \cite{Gesbert}, and whose sum degrees of freedom (DoF) has been derived in \cite{Gesbert} through interference avoidance based on fractional graph coloring. The work in \cite{MutanganaTandonICC20} derived a lower bound on sum SDoF for the \textit{half-rate-feasible} TIM-CM problem. This is the secrecy constrained version of the well understood \textit{half-rate-feasible} TIM  problem whose upper bound on symmetric DoF has been shown in \cite{Jafar, MalekiJafarCadambe} to be $1/2$ per user, hence the name ``half-rate-feasible."  

Just as there is a direct relationship between the general TIM problem and the general ICOD problem, one can also reformulate the TIM-CM problem of the current paper into an \textit{ICOD problem with confidential messages} (private-ICOD) and vice versa. The main difference between TIM-CM and private-ICOD is that, unlike in TIM-CM, each receiver in private-ICOD has an arbitrary subset of transmitted messages and jamming signals as side information, i.e., requiring that each receiver be only able to decode its intended message and nothing more beyond its side information set. Under the private-ICOD settings, all messages and jamming signals (from the transmitters) are assumed to be available to the common node which encodes them into a common function to broadcast to all receivers. Each receiver can use its side information to decode its indented message. As a consequence of the possible reformulation,  we also consider the literature on  secrecy constrained index coding, e.g.,  \cite{DauSkachekChee2012, MojahedianGohariAref2015, OngVellambiYeohKliewerYuan2016, OngKliewerVellambi2018, LiuVellambiKimSadeghi2018, LiuSadeghiAboutorabSharififar2020, NarayananPrabhakaranRaviDeyKaramchandani2018, NarayananPrabhakaranRaviDeyKaramchandani2020, SasiRajan2019, LiuTuninetti2019, LiuTuninetti2020}, to get a comprehensive picture of works related to TIM-CM.

The majority of results on secrecy constrained ICOD problems can be broadly placed in two major categories: $(i)$ ICOD problems with secrecy constraints against external eavesdroppers, i.e.,  requiring secrecy against any illegitimate network users that may eavesdrop on the broadcast function. This is a problem also referred to as secure index coding (secure-ICOD), e.g., \cite{MojahedianGohariAref2015, OngVellambiYeohKliewerYuan2016, OngKliewerVellambi2018, LiuVellambiKimSadeghi2018}. $(ii)$ ICOD problems with secrecy constraints against other legitimate users within the same network. This second category can be further subdivided into two problem subclasses, namely, private-ICOD (also defined above) and private pliable index coding (private-PICOD), which focuses on a slightly more relaxed setting where each receiver is satisfied by decoding any one arbitrary message from the transmitted messages and nothing more beyond its side information set. For example, in recent work \cite{LiuVellambiKimSadeghi2018}, the authors studied the $K$-user secure-ICOD problem and characterized the optimal capacity region for $K\leq 4$. On the other hand, in a recent paper \cite{NarayananPrabhakaranRaviDeyKaramchandani2018} and its long version \cite{NarayananPrabhakaranRaviDeyKaramchandani2020}, the authors studied the $K$-user private-ICOD problem and settled the capacity region for $K\leq 3$.  In order to obtain this result,  there is an additional assumption made in \cite{NarayananPrabhakaranRaviDeyKaramchandani2018}: a key sharing mechanism enables the common node to share secret keys with arbitrarily subsets of receivers in addition to the broadcast of the common function. Moreover, it is proved in \cite{NarayananPrabhakaranRaviDeyKaramchandani2018} that, due to the presumption of existence of this secret key sharing mechanism,  the achievable secrecy rate matches that of the nonsecrecy constrained ICOD problem, i.e., there is no rate penalty for secrecy.  In the absence of the secrect key sharing mechanism between the server and receivers,  the authors \cite{NarayananPrabhakaranRaviDeyKaramchandani2018} also studied a resulting \textit{weak secrecy private-ICOD} (WS-ICOD) model and provided necessary conditions for the feasibility of positive secrecy rate. For this model, no achievable schemes were proposed and capacity region was not derived for any value of $K$. A direct reformulation of the TIM-CM problem of the current paper would lead to a WS-ICOD model because WS-ICOD precludes the assumption of the secret key sharing mechanism beyond the broadcast function. Moreover, there is a rate penalty for secrecy under TIM-CM settings. In regards to the private-PICOD problem, recent works derived lower bounds on secrecy rate \cite{SasiRajan2019, LiuTuninetti2019, LiuTuninetti2020}.

Motivated by the above discussion, in this paper, we focus on the following question: What are the bounds on symmetric SDoF for the $K$-user partially connected interference networks with confidential messages in the absence of CSIT, except for the knowledge of network topology and channel statistics? In other words, our aim is to explore how much this topology knowledge alone can be exploited in order to manage interference and establish limits on symmetric SDoF for the TIM-CM problem. 

\indent\textbf{Contributions:}
We summarize our main contributions as follows:
\begin{enumerate}[(i)]
\item \textit{(Necessary and sufficient conditions for symmetric SDoF feasibility):} We characterize necessary and sufficient conditions under which non-zero SDoF is feasible for any topology. These feasibility conditions for TIM-CM are related to conditions of its WS-ICOD counterpart \cite [Lemma 8]{NarayananPrabhakaranRaviDeyKaramchandani2020} since its graph is a complementary of the network coding graph of WS-ICOD. 

\item \textit{(Inner bounds on symmetric SDoF):}  We propose achievable symmetric SDoF schemes for the general TIM-CM problem by  introducing two transmission schemes, namely, \textit{secure partition} for the first lower bound and \textit{secure independent} sets for the second one. 

\item \textit{(Outer bounds on symmetric SDoF):}  We also obtain upper bounds on symmetric SDoF for the general TIM-CM problem. To this end, we first show how to obtain a nontrivial upper bound on symmetric SDoF through a careful analysis of the received signal structures at all $K$ receivers and their respective interference components with regard to the underlying topology.  We next present a second upper bound which we show to be tighter for some examples. The main difference between this upper bound and the first one is that here we leverage the potential presence of \textit{fractional signal generators} within the network, a concept that was introduced in  \cite{NaderializadehAvestimehr, Gesbert} for the nonsecrecy constrained TIM problem. This is a paradigm where the interference signal component of the received signal at some receiver can generate either a statistically equivalent version of the received signal at some other receiver or its cleaner version (i.e., with less interference). 

\item \textit{(Optimal symmetric SDoF for all $K$-user TIM-CM topologies with $K\leq 3$):} Finally, we apply the proposed upper and lower bounds to characterize the optimal symmetric SDoF for all $K$-user TIM-CM topologies with $K\leq 3$. 
\end{enumerate}

\section{System Model} \label{SystemModel}

\begin{figure*}[!t]
	\begin{center}
		\includegraphics[width=0.7\textwidth]{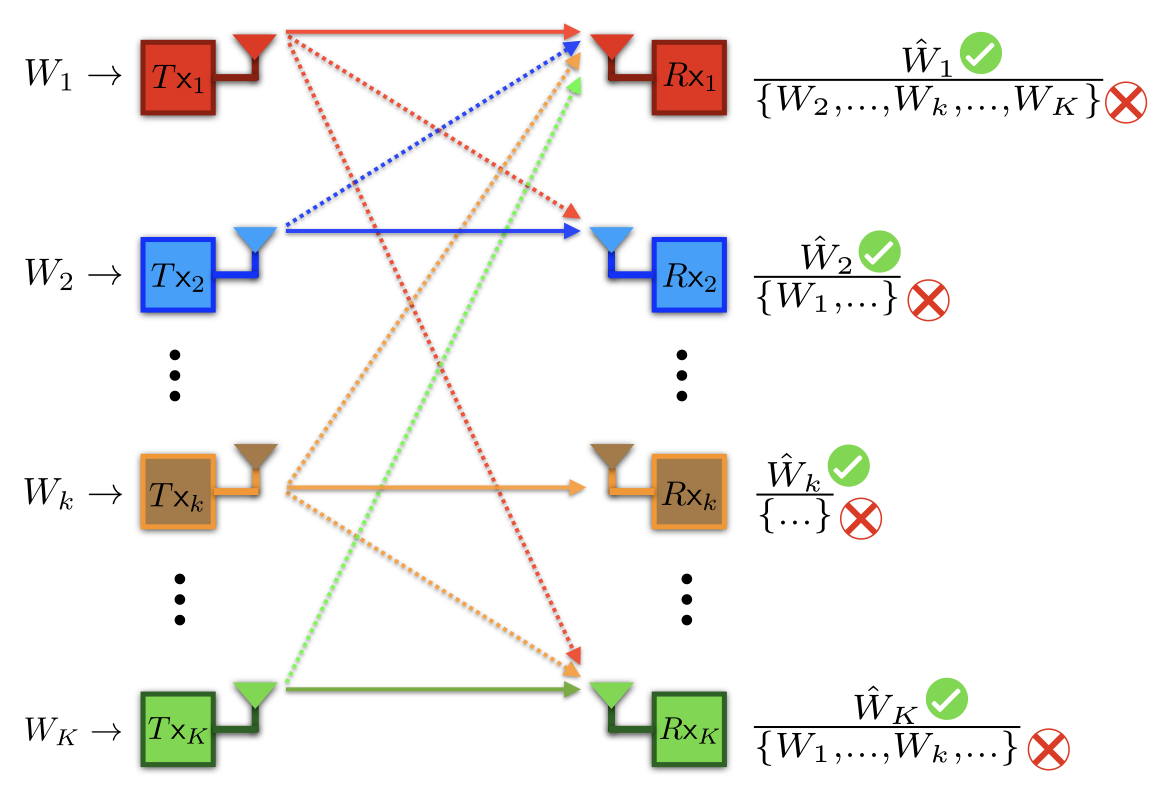}
	\vspace{2pt}
		\caption{TIM-CM problem where each user $R\textsf{x}_k$ should only decode its intended message $W_k$ and the remaining (interfering) messages must remain confidential.}
		\label{fig:Fig1}
	\end{center}
\end{figure*} 
We consider the $K$-user partially connected single-input single-output (SISO)  interference network with confidential messages as depicted in Fig. \ref{fig:Fig1}. We use $\mathcal{T}_k$ to denote the index set of all transmitters that are connected to receiver $R\textsf{x}_k$, for $k\in\{1, 2, \dots, K\}$. Similarly,  we use $\mathcal{R}_k$ to denote the index set of all receivers that are connected to transmitter $T\textsf{x}_k$, for $k\in\{1, 2, \dots, K\}$. Therefore, the network can be fully described by its topology $\mathcal{G}=(\mathcal{T}_1,\mathcal{T}_2, \dots, \mathcal{T}_K,  \mathcal{R}_1, \mathcal{R}_2, \dots, \mathcal{R}_K)$. Additionally, we use $\mathcal{I}_k$ to represent the index set of all transmitters that cause interference at $R\textsf{x}_k$.  In turn, this implies that $\mathcal{T}_k=\{k\}\cup \mathcal{I}_k$. For a given topology $\mathcal{G}$, the signal received at $R\textsf{x}_k$ at time $t$ is given by
\begin{align}
\label{SystemModelEquation}
Y_k(t)=h_{kk}X_k(t)+\sum_{i\in\mathcal{I}_k}h_{ki}X_i(t)+Z_k(t),
\end{align}
where $h_{ki}$ represents the channel coefficient between Transmitter $i$ and Receiver $k$, assumed be time invariant and i.i.d. across users. $Z_k(t)$ is the zero-mean unit-variance complex Gaussian channel noise. We next provide the  adjacency matrix definition which will be useful later.
\begin{definition}(Adjacency matrix): For a network topology $\mathcal{G}$, we define an adjacency matrix $\mathbf{B}$ as:
\[ \mathbf{B}_{(j,i)}=\begin{cases} 
      1, & h_{ji}\neq 0 \\
      0, & {otherwise},
   \end{cases}
\]
where the matrix $\mathbf{B}$ is of size $K\times K$. Here, $h_{ji}\neq 0$ ($h_{ji}=0$) indicates whether there is a connection (no connection) between Transmitter $i$ and Receiver $j$ \cite{Jafar, Gesbert}.
\end{definition} 

\textbf{CSIT/CSIR assumptions}: Under the TIM-CM framework, there is no CSIT, except that the transmitters have the statistics of the channel coefficients and full access to the network topology $\mathcal{G}$.  For coherent signal detection, in addition to having access to $\mathcal{G}$,  we assume that there is global CSI at the receivers (CSIR). That is, each receiver $R\textsf{x}_k$, for $k\in\{1, 2, \dots, K\}$, has casual access to the set of all non-zero channel coefficients denoted by $\mathcal{H}=\{h_{ki}:i \in \mathcal{T}_k\}_{k=1}^K$. 


Each transmitter  $T\textsf{x}_k$ , for $k\in\{1, 2, \dots, K\}$, uses the knowledge of  $\mathcal{G}$ to encode its message $W_k$ into an $n$-length vector $X^n_k=[X_k(1), X_k(2), \dots, X_k(n)]^\top$  via an encoding function $f_{_{Enc}}(W_k|\mathcal{G})$. The vector  $X^n_k$ is subject to the following power constraint:
\begin{align}
\label{PowerConstraint}
\frac{1}{n}\mathbb{E}(||X^n_k||^2)\leq P,
\end{align}
 where $P$ is the average transmit power. Each receiver $R\textsf{x}_k$ receives a vector  $Y^n_k=[Y_k(1), Y_k(2), \dots, Y_k(n)]^\top$ and uses the knowledge of both $\mathcal{G}$ and $\mathcal{H}$ to recover its intended message $W_k$ via a decoding function $f_{_{Dec}}(Y^n_k|\mathcal{G}, \mathcal{H})$.  The considered system model  is subject to reliability and secrecy constraints as we explain next.

Transmitter $T\textsf{x}_k$ wants to securely send a message $W_k$, which is uniformly distributed in $\mathcal{W}_k=\{1, 2, \dots, 2^{nR_k}\}$, to receiver $R\textsf{x}_k$. Here, $\mathcal{W}_k$ represents the index set of all message at $T\textsf{x}_k$. A secure rate of communication $R_k(P, \mathcal{G})=\frac{\log(|\mathcal{W}_k|)}{n}$ is achievable, if there exists a sequence of encoding and decoding functions $f_{_{Enc}}(W_k|\mathcal{G})$, $f_{_{Dec}}(Y^n_k|\mathcal{G}, \mathcal{H})$ such that, as $n\rightarrow\infty$, both the decodability and confidentiality constraints are satisfied:

\noindent\textbf{Decodability constraint:} Each receiver $R\textsf{x}_k$ must be able to decode its intended message $W_k$,
\begin{align}
\label{DecodabilityConstraint}
&Pr[W_k\neq \hat{W}_k]= o(n).
\end{align}
\textbf{Confidentiality constraints:}   All messages seen as interference at receiver $R\textsf{x}_k$ must remain confidential,
\begin{align}
\label{ConfidentialityConstraint}
&\frac{1}{n}I(W_{\mathcal{I}_k}; Y_{k}^{n}|W_k,\mathcal{H})= o(n), \quad \forall k\in\{1, 2, \dots, K\}.
\end{align}
Here, $W_{\mathcal{I}_k}=\{W_i:i\in\mathcal{I}_k, i\neq k\}$ is the set of all messages seen as interference at $R\textsf{x}_k$, and $Y_{k}^{n}$ is the signal received at $R\textsf{x}_k$ over the $n$-length transmission block.

\textcolor{black}{\begin{definition} \label{SDoFDefinition}(Secure degrees of freedom (SDoF) tuple): Consider a network topology $\mathcal{G}$ and the defined above average signal transmission power parameter $P$. We say that an SDoF $K$-tuple $\left(\textup{\textsf{SDoF}}^{(1)}, \textup{\textsf{SDoF}}^{(2)}, \dots, \textup{\textsf{SDoF}}^{(K)}\right)$ is achievable, if there exist achievable rates $\left(R_1(P, \mathcal{G}), R_2(P, \mathcal{G}), \ldots, R_K(P, \mathcal{G})\right)$, for each $P$,  such that 
\begin{align}
\label{SDoFPrelogEquation}
\textup{\textsf{SDoF}}^{(k)}=\lim_{P\to\infty}\frac{R_k(P, \mathcal{G})}{\log(P)}, \quad k\in\{1, 2, \dots, K\},
\end{align}
where $\textup{\textsf{SDoF}}^{(k)}$ is the achievable SDoF for user $k$. 
\end{definition}}

In this paper, we focus on the \textit{symmetric SDoF} which we define next.

\begin{definition}\label{SDoF_symDefinition}(Symmetric SDoF): The symmetric SDoF, $\textup{\textsf{SDoF}}^{\textsf{\textup{sym}}}$ is defined as follows: 
\begin{align}
&\textup{\textsf{SDoF}}^{\textsf{\textup{sym}}}\nonumber\\&\qquad=\sup\big\{ D:\left(\textup{\textsf{SDoF}}^{(1)}, \textup{\textsf{SDoF}}^{(2)}, \dots, \textup{\textsf{SDoF}}^{(K)}\right)\nonumber\\&\qquad=(D, \dots, D)\text{~~is achievable.}\big\}
\end{align}
\end{definition}

\section{Main Results and Discussion}

In this Section, we present our main results and an accompanying discussion together with illustrative examples. In Section \ref{SectionSDOFFeasibility andInnerBounds}, we present necessary and sufficient conditions for the feasibility of positive symmetric SDoF for the general TIM-CM problem, the first inner bound on symmetric SDoF based on secure partition, and the second inner bound based on secure independent sets, respectively in Theorems \ref{TheoremSDoFSymmetry}, \ref{TheoremGeneralTIM-CM-SecurePartition}, and \ref{TheoremGeneralTIM-CMSDoFNonGreedy}. In Section \ref{SectionOuterBoundsonSymmetricSDoF}, we present our first and second upper bounds on symmetric SDoF, respectively in Theorems \ref{Theorem4UpperBound1} and \ref{Theorem5UpperBound2}. Finally, in Section \ref{SectionCaseStudies}, we present two case studies where we respectively characterize the optimal symmetric SDoF for all $2$-user and  $3$-user TIM-CM topologies.

\subsection{Feasibility of Positive Symmetric SDoF and Inner Bounds}\label{SectionSDOFFeasibility andInnerBounds}
\label{SymmetricSDoFFeasibilitySection}


Consider the $K$-user TIM-CM model defined by equation \eqref{SystemModelEquation}.  The following Theorem, which is proved in Appendix \ref{appendic:Theorem1Proof}, answers the question: When is  symmetric SDoF zero for a given TIM-CM network? 

\begin{theorem}\label{TheoremSDoFSymmetry}
For a network topology $\mathcal{G}$, the symmetric SDoF is zero if and only if there exists a pair $i, j\in\{1, 2, \dots, K\}$, for $i\neq j$, such that
\begin{enumerate}[(i)]
\item $i\in\mathcal{I}_{j}$
\item $j \in\mathcal{I}_{i}$ 
\item $\mathcal{I}_{j}\setminus \{i\} \subseteq \mathcal{I}_{i}\setminus{\{j\}}$ or  $\mathcal{I}_{i}\setminus{\{j\}} \subseteq \mathcal{I}_{j}\setminus{\{i\}}$.
\end{enumerate}
\end{theorem}

Consequently, as also detailed in Appendix \ref{appendic:Theorem1Proof}, for any TIM-CM network that does not satisfy the conditions of Theorem \ref{TheoremSDoFSymmetry}, positive symmetric SDoF is feasible. In the Appendix, we present a simple scheme satisfying positive symmetric SDoF through \textit{secure} time division multiple access (secure TDMA), which leads to the smallest achievable symmetric SDoF of $\frac{1}{K}$ per user. We use Examples \ref{Example1} and  \ref{Example2} to highlight the principles behind Theorem \ref{TheoremSDoFSymmetry}.

\begin{figure*}[!t]
        \captionsetup[subfigure]{aboveskip=-1pt,belowskip=-1pt}
        \centering
        \begin{subfigure}[b]{0.49\textwidth}
                \includegraphics[width=\textwidth]{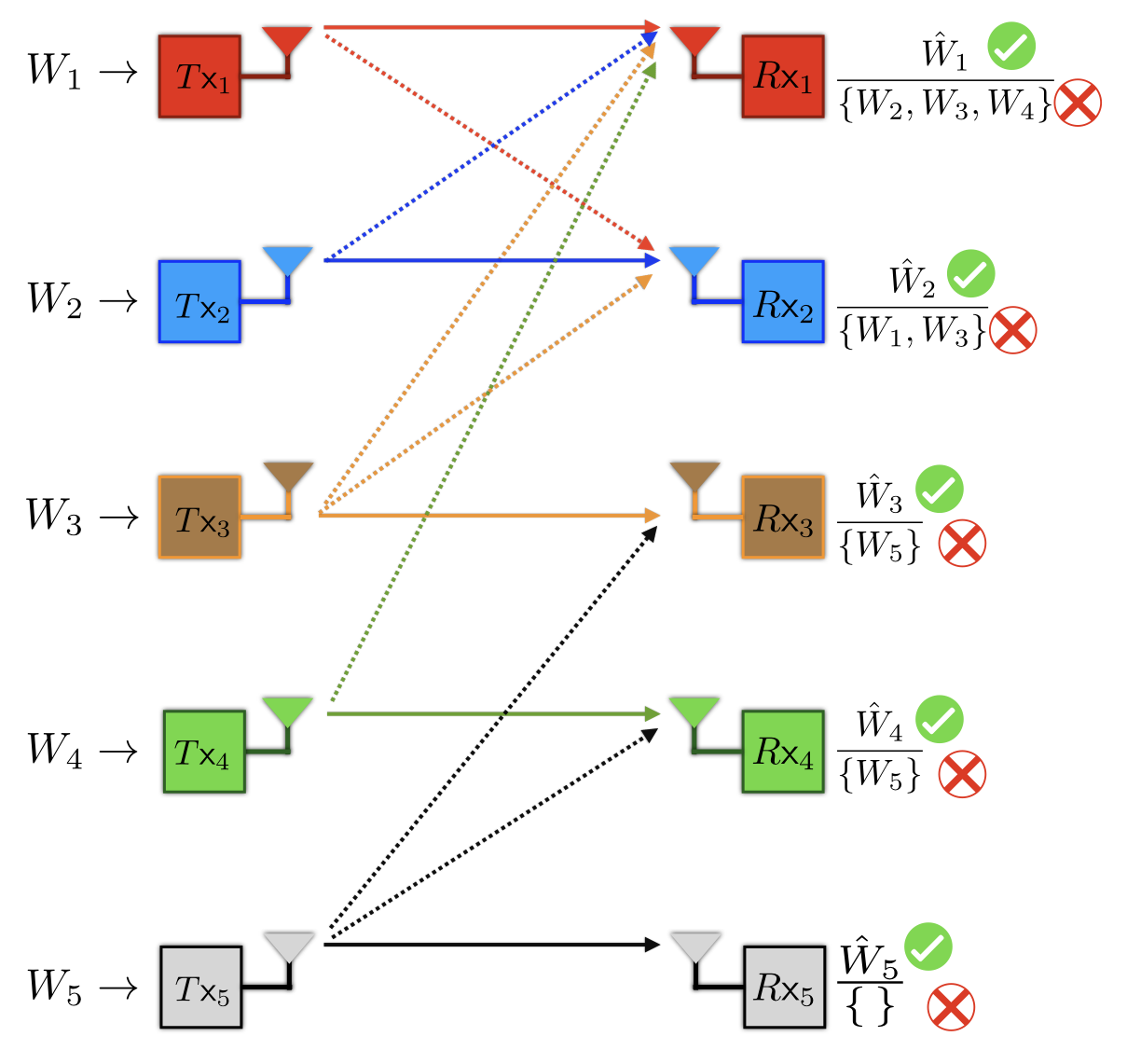}
                \caption{}
                 \label{a-ZeroSDoF}
        \end{subfigure} 
        \begin{subfigure}[b]{0.49\textwidth}
                \includegraphics[width=\textwidth]{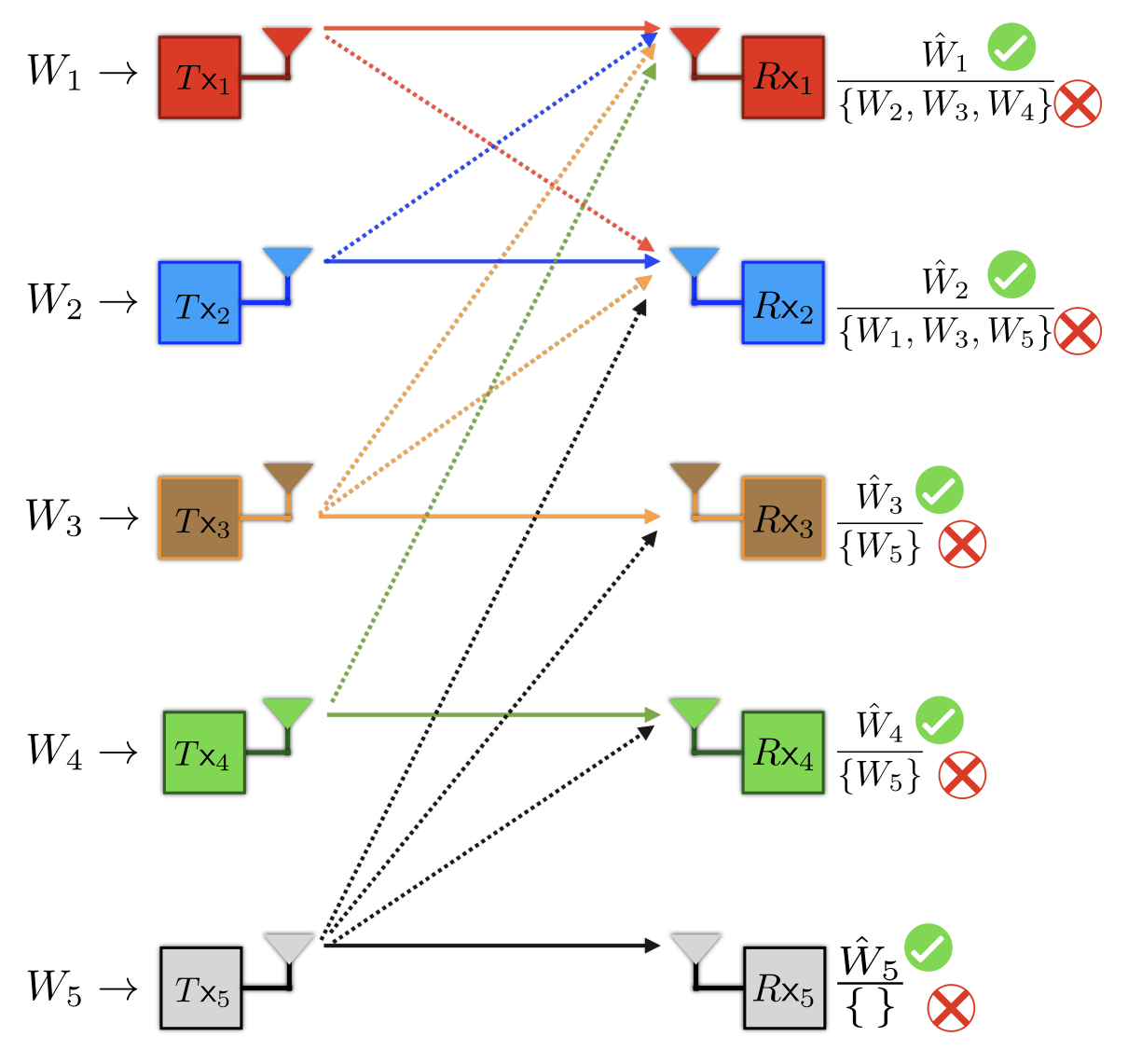}
                \caption{}
                \label{b-NonZeroSDoF}
        \end{subfigure}
        \vspace{5pt}
        \caption{$5$-user TIM-CM networks depicting: (a) When symmetric SDoF is \textit{zero} versus (b) When it is \textit{nonzero}.}
        \label{fig:FigSym}
\end{figure*}
\begin{manualexample}{1}  \label{Example1}($5$-user TIM-CM network: Zero symmetric SDoF):  \normalfont
Consider the network topology in Fig. \ref{fig:FigSym} \subref{a-ZeroSDoF}. At receiver $R\textsf{x}_1$, we have $\mathcal{I}_1=\{2, 3, 4\}$ and $\mathcal{T}_1=\{1\}\cup \mathcal{I}_1=\{1, 2, 3, 4\}$. At receiver $R\textsf{x}_2$, we have $\mathcal{I}_2=\{1, 3\}$ and $\mathcal{T}_2=\{2\}\cup \mathcal{I}_2=\{1, 2, 3\}$. Consider the user pair $(i, j)=(1, 2)$. Therefore, for this pair,  the conditions of Theorem \ref{TheoremSDoFSymmetry} are satisfied because $(i)$ $1\in\mathcal{I}_2$, $(ii)$ $2\in\mathcal{I}_1$, and $(iii)$ $\mathcal{I}_2\setminus\{1\}\subseteq\mathcal{I}_1\setminus\{2\}$.  Since the conditions of Theorem \ref{TheoremSDoFSymmetry} are satisfied, then symmetric SDoF is zero for this network. 

The intuition behind the infeasibility of positive symmetric SDoF for this topology is as follows: Suppose transmitter $T\textsf{x}_1$ wants to send its message during a given time slot. Then, transmitters $T\textsf{x}_2$, $T\textsf{x}_3$, and $T\textsf{x}_4$, which are all seen at its intended receiver $R\textsf{x}_1$ have to remain silent in order to avoid causing interference and thus preventing it from decoding the message sent by $T\textsf{x}_1$. Moreover, any message sent by $T\textsf{x}_1$  will be seen as interference at $R\textsf{x}_2$. Therefore, a separate transmitter that is seen at $R\textsf{x}_2$ but not seen at $R\textsf{x}_1$ is needed to protect the message from  $T\textsf{x}_1$ by jamming  $R\textsf{x}_2$. However, there is no such transmitter within the network topology of Fig. \ref{fig:FigSym} \subref{a-ZeroSDoF}, thereby leading to zero secure communication rate for the transmit receive pair  $T\textsf{x}_1-R\textsf{x}_1$. The proof of Theorem \ref{TheoremSDoFSymmetry} essentially formalizes the above logical arguments.
\end{manualexample}

\begin{manualexample}{2}  \label{Example2}($5$-user TIM-CM network: Nonzero symmetric SDoF): \normalfont 
Consider the network topology in Fig. \ref{fig:FigSym} \subref{b-NonZeroSDoF}. It can easily be verified that for all user pairs $i, j\in\{1, 2, \dots, 5\}$, for $i\neq j$,  the conditions of Theorem \ref{TheoremSDoFSymmetry} are not satisfied. Therefore, we can achieve a positive symmetric SDoF of at least $\frac{1}{K}=\frac{1}{5}$ over five time slots as we show next. 

The proposed transmission works over $T=5$ \textit{sub-blocks}, where each sub-block is of length $n_B$. This leads to the transmission block of length $n=n_BT$.  Moreover, we assume that $n_B$ is large enough to satisfy the decodability and confidentiality constraints in \eqref{DecodabilityConstraint}-\eqref{ConfidentialityConstraint}. From here onward in the achievable schemes of this paper, with a slight abuse of notation, we refer to a sub-block of length $n_B$ as a \textit{time slot}. Under the same convention, we will refer to $T$ as the \textit{transmission blocklength}. Furthermore, we assume that the transmitted signals are observed at discrete and synchronous time slots at the receivers.

\begin{enumerate}[(i)]
\item In the first sub-block (time slot), transmitter $T\textsf{x}_1$ sends its message signal using a wiretap code \cite{Wyner}, whereas $T\textsf{x}_5$ acts as a cooperative jammer by sending an artificial noise signal. All the other transmitters remain silent. Hence, the received signals at the first and second receivers are respectively given by $Y_1=h_{11}X_1+Z_1$ and $Y_2=h_{21}X_1+h_{25}X_5+Z_2$. The following secrecy rate is achievable at $R\textsf{x}_1$: $R_1=I(X_1; Y_1|\mathcal{H})-I(X_1; Y_2|\mathcal{H})$. In particular, by assuming that the transmitters use Gaussian codebooks \cite{LiPetropulu2011, SLinPLin2013, LinLin, TYLiuUlukus}, it can be readily verified that the achievable SDoF at $R\textsf{x}_1$ is $\textup{\textsf{SDoF}}^{(1)}= \lim_{P\to\infty} \frac{R_1}{\log(P)}=1$. Transmission for the remaining four users follows the same logic as the first one and we summarize it next. 

\item In time slot $t=2$, let $T\textsf{x}_2$ transmit its message and let $T\textsf{x}_4$ jam $R\textsf{x}_1$.
\item In time slot $t=3$, let $T\textsf{x}_3$ transmit  its message and let $T\textsf{x}_1$ simultaneously jam $R\textsf{x}_1$ and  $R\textsf{x}_2$.
\item In time slot $t=4$, let $T\textsf{x}_4$ transmit its message and let $T\textsf{x}_2$ jam $R\textsf{x}_1$.
\item In time slot $t=5$, let  $T\textsf{x}_5$ transmit its message and let $T\textsf{x}_3$ and $T\textsf{x}_4$ collectively jam $R\textsf{x}_2$, $R\textsf{x}_3$, and $R\textsf{x}_4$. Hence, each receiver gets its intended message over five time slots, i.e., $\textup{\textsf{SDoF}}^{\textsf{\textup{sym}}}\geq \frac{1}{5}$.
\end{enumerate}
\end{manualexample}

We next present general transmission schemes beyond the (greedy) secure TDMA approach. We show how to achieve lower bounds on symmetric SDoF that are greater than or equal to $\frac{1}{K}$ under feasible topology conditions. Clearly, this implies that, under such schemes, the transmission blocklength may be less than the number of users in the whole network. As we explain later, transmission is done through the process of \textit{secure interference avoidance}. Let us first define  \textit{interference avoidance}.

\begin{definition} (Interference avoidance): In order to enable decodability at its intended receiver, each transmitter $T\textsf{x}_{k}$, for $k\in \{1, 2, \dots, K\}$, uses time slot $t$, for $t\in \{1, 2, \dots, T\}$, to transmit an information message $W_k$, for $k\in \{1, 2, \dots, K\}$, if and only if both those users that it causes interference to and those users that cause interference at its intended receiver $R\textsf{x}_{k}$ are not using the same time slot $t$ to transmit any signals. This principle is called  interference avoidance.
\end{definition}

The above transmission principle would suffice, if we were only concerned with decodability, i.e., under the TIM model. However, for the TIM-CM problem, we are additionally required to preserve secrecy. Thus, more restrictions on transmission have to be imposed beyond interference avoidance as we explain next.

\begin{definition} (Secure interference avoidance):  In order to enable both decodability and secrecy, $T\textsf{x}_{k}$ can transmit its message $W_k$ by following the defined above interference avoidance plus a new requirement that there must exist a set $\mathcal{C}$ of cooperative jamming transmitters (that do not interfere at $R\textsf{x}_{k}$) to protect $W_k$ at all receive nodes where it is seen as interference by sending artificial noise.
\end{definition}

The above described (secure) interference avoidance principles are closely related to the notion of (secure) \textit{independent sets}  and \textit{secure partition} as we explain next. 

\begin{definition} \label{IndependentSet}(Independent set (IS)):  Consider a $K$-user TIM-CM network. A set of users $\mathcal{U}_{IS} \subseteq\{1, 2, \dots, K\}$, is an \textit{independent set}, if for all $i\neq j \in \mathcal{U}_{IS}$,  $\mathbf{B}_{(i, j)}=\mathbf{B}_{(j, i)}=0$. This implies that,  for all $i\neq j\in \mathcal{U}_{IS}$, transmitter $T\textsf{x}_{j}$ is not connected to receiver $R\textsf{x}_{i}$ and transmitter $T\textsf{x}_{i}$ is not connected to receiver $R\textsf{x}_{j}$. This can simply be represented as $j\notin \mathcal{T}_i$ and $i\notin \mathcal{T}_j$.
\end{definition}
\begin{definition}\label{SecureIndependentSet} (Secure independent set (SIS)):  A set of users $\mathcal{U}_{SIS} \subseteq\{1, 2, \dots, K\}$, is a \textit{secure independent set}, if $\mathcal{U}_{SIS}$ is an independent set, and there exists
a set $\mathcal{C} \subseteq\{1, 2, \dots, K\}\setminus \mathcal{U}_{SIS}$, such that:
\begin{enumerate}[(i)]
\item $\mathcal{U}_{SIS}\cap\mathcal{R}_\mathcal{C}=\emptyset$ 
\item $\mathcal{R}_{\mathcal{U}_{SIS}}\setminus \mathcal{U}_{SIS}\subseteq \mathcal{R}_\mathcal{C}$,
\end{enumerate}
\textcolor{black}{where $\mathcal{R}_{\mathcal{U}_{SIS}}$ and $\mathcal{R}_\mathcal{C}$ respectively represent the set of all receivers that are connected to transmitters in $\mathcal{U}_{SIS}$ and $\mathcal{C}$.}

In other words,  condition $(i)$ implies that the intended receivers for the transmitters in $\mathcal{U}_{SIS}$ do not see the cooperative jamming signals from the transmitters in $\mathcal{C}$; and condition $(ii)$ implies that all unintended receivers that see signals from the transmitters in $\mathcal{U}_{SIS}$  should be in the coverage of the cooperative jamming set $\mathcal{C}$. 
\end{definition}

\textcolor{black}{The intuition behind the above definition is that each transmitter in the secure independent set  $\mathcal{U}_{SIS}$ will also be able to achieve confidentiality, as it will be simultaneously protected by artificial noise signal(s) that are sent by the cooperative jamming transmitters  in $\mathcal{C}$.}

\textcolor{black}{\begin{definition} \label{SecurePartition}(Secure partition (SP)): 
A partition $\mathcal{P}= \{\mathcal{P}_1, \mathcal{P}_2, \ldots, \mathcal{P}_{T}\}$, where $\mathcal{P}_t\neq \emptyset$ for all $t\in\{1, 2, \dots, T\}$, $\mathcal{P}_m \cap \mathcal{P}_n \neq \emptyset$ for all $m\neq n \in\{1, 2, \dots, T\}$, and $\bigcup_{t=1}^{T} \mathcal{P}_t = \{1, 2, \dots, K\}$,  is a secure partition, if every subset $\mathcal{P}_t$, for $t=1, 2, \ldots, T$,  is a secure independent set.
\end{definition}}

The following Theorem states our first lower bound on the symmetric SDoF.
\begin{theorem}\label{TheoremGeneralTIM-CM-SecurePartition}
The symmetric SDoF for the $K$-user TIM-CM network is lower bounded as follows:
\begin{align}
&\textup{\textsf{SDoF}}^{\textup{\textsf{sym}}}\geq \max_{\mathcal{P}} \frac{1}{|\mathcal{P}|},\nonumber\\
&\text{such that~} \mathcal{P}= \{\mathcal{P}_1, \mathcal{P}_2, \ldots, \mathcal{P}_{T}\}, \text{~for~} |\mathcal{P}|=T,\nonumber\\& \text{~is a secure partition of~} \{1, 2, \dots, K\}.
\end{align}
\end{theorem}

\begin{proof}
To prove the above result, we need to argue that for any secure partition $\mathcal{P}$, a symmetric SDoF of $1/|\mathcal{P}|$ is achievable. The result then follows by maximizing over all secure partitions. Since by Definition \ref{SecurePartition}, each subset $\mathcal{P}_t$ is a secure independent set, we can schedule all users in $\mathcal{P}_t$ in a single time slot $t$, for $t\in\{1, 2, \dots, T\}$,   while satisfying decodability \eqref{DecodabilityConstraint} and secrecy \eqref{ConfidentialityConstraint}. To this end, we also  simultaneously schedule all transmitters in a separate set $\mathcal{C}_t$, for $t\in\{1, 2, \dots, T\}$, to act as cooperative jammers by transmitting artificial noise symbols during time slot $t$ in order to protect all the information signals from the $|\mathcal{P}_t|$ transmitters at the receivers where they are seen as interference.  Thus, by separately scheduling  $|\mathcal{P}|$ subsets over $|\mathcal{P}|$ slots, and the fact a user only appears in a single subset, we achieve a lower bound of $1/|\mathcal{P}|$. Clearly, under this transmission scheme, each user is only assigned a single secure interference free channel use per transmission blocklength $|\mathcal{P}|=T$ and the achievable SDoF is maximized when all receivers can be served with their intended messages over the minimum cardinality of $\mathcal{P}$.
\end{proof}

\begin{figure*}[!t]
        \captionsetup[subfigure]{aboveskip=-10pt,belowskip=-1pt}
        \centering
        \begin{subfigure}[b]{0.65\textwidth}
                \includegraphics[width=\textwidth]{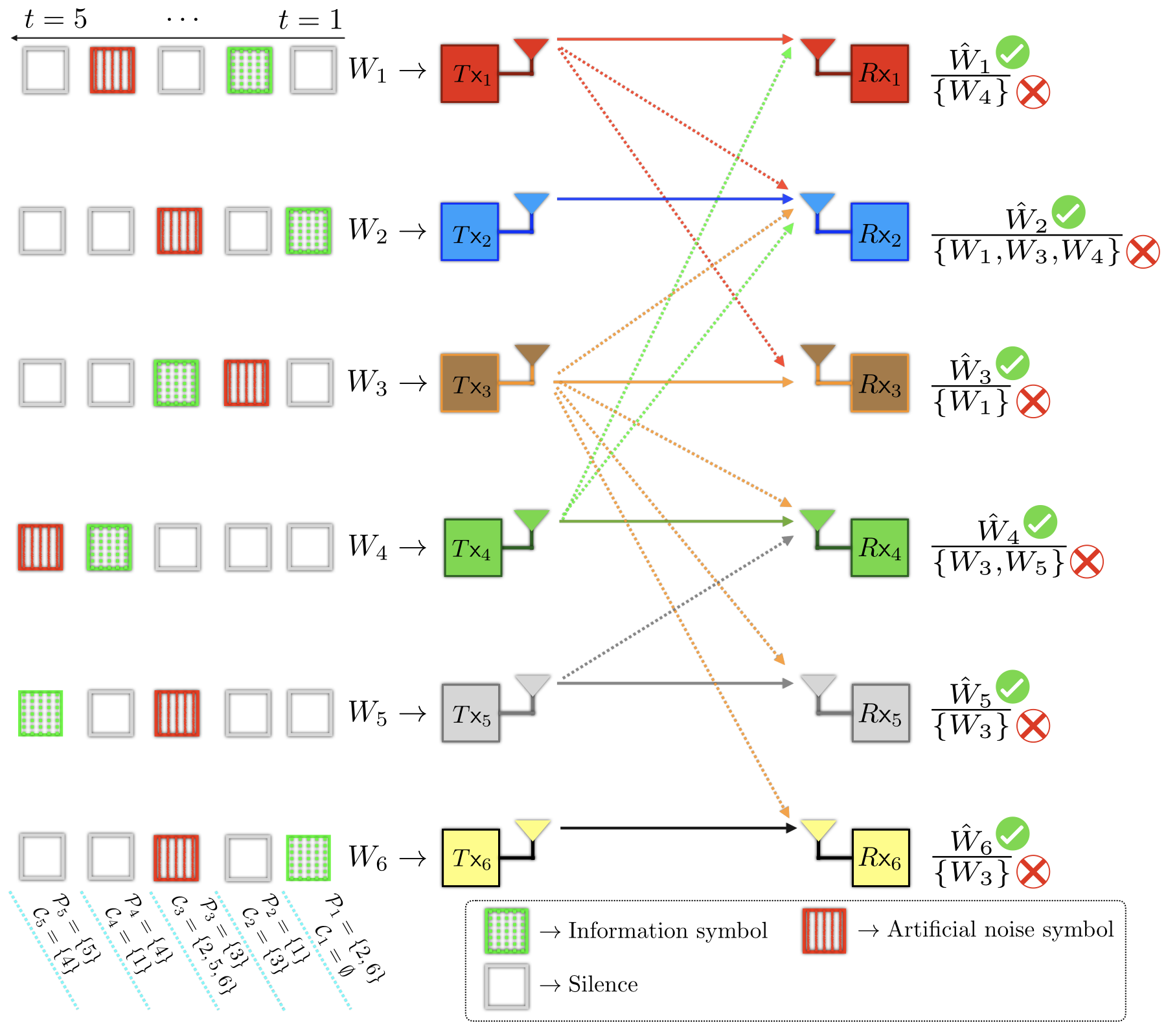}
                \vspace{3pt}
                  \caption{}
                \label{a-SecurePartition1}
        \end{subfigure}      
        \quad 
        \begin{subfigure}[b]{0.7\textwidth}
                \includegraphics[width=\textwidth]{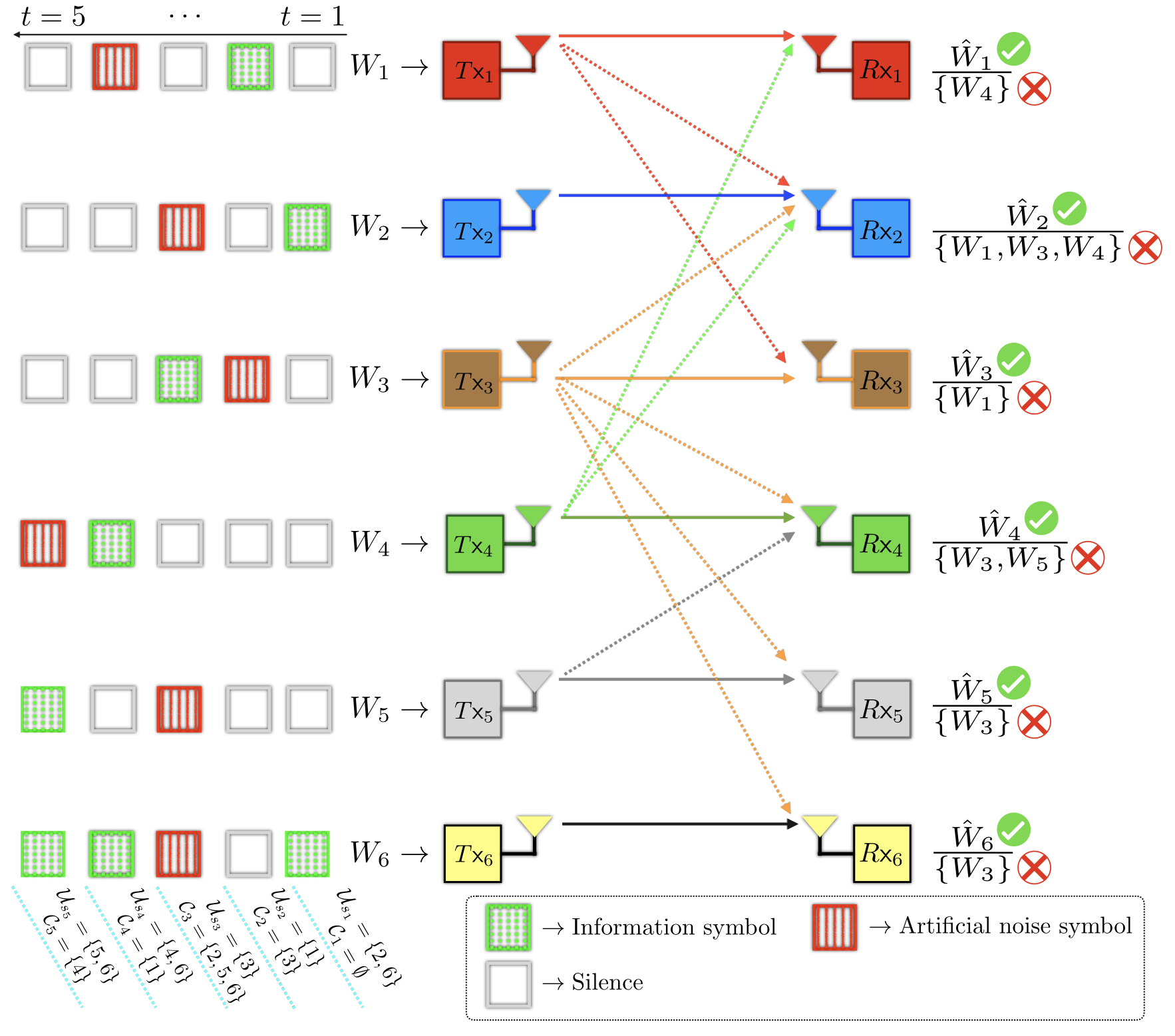}
                \vspace{3pt}
                \caption{}
                \label{b-SecureMaximalIndependentSets1}
        \end{subfigure}
        \vspace{5pt}
        \caption{$6$-user TIM-CM network contrasting: (a) Transmission via secure partition  versus  (b) Transmission via {\color{black}secure independent sets}.}
        \label{fig:Fig2}
\end{figure*}

We use Examples \ref{3.A} and \ref{4.A} to highlight the principles behind  Theorem \ref{TheoremGeneralTIM-CM-SecurePartition}.
\begin{manualexample}{3.A}  \label{3.A} ($6$-user TIM-CM network: Transmission via SP): \normalfont Consider the network topology $\mathcal{G}=(\mathcal{T}_1,\mathcal{T}_2, \dots, \mathcal{T}_6,  \mathcal{R}_1, \mathcal{R}_2, \dots, \mathcal{R}_6)$ depicted in Fig. \ref{fig:Fig2}. From this topology, we can directly deduce the following:
\begin{itemize}
\item \textit{(Independent sets)}: Consider a set of transmitters $\mathcal{U}_{IS}=\{1, 5, 6\}$. It can be verified that, for all $i, j\in\mathcal{U}_{IS}$, the adjacency matrix of this network satisfies $\mathbf{B}_{(i, j)}=\mathbf{B}_{(j, i)}=0$. Therefore, by Definition \ref{IndependentSet},  Transmitters $1$, $5$, and $6$ form an independent set $\{1, 5, 6\}$. In other words, $T\textsf{x}_1$, $T\textsf{x}_5$, and $T\textsf{x}_6$ are not connected to each other's  desired destinations $R\textsf{x}_1$, $R\textsf{x}_5$, and $R\textsf{x}_6$. Similarly, we can show that the sets $\{3\}$, $\{4, 6\}$, and $\{2, 5, 6\}$ are also independent sets.

\item \textit{(Secure independent sets)}: Consider a set of transmitters $\mathcal{U}_{SIS}=\{5, 6\}$. Clearly, for all $i, j\in\mathcal{U}_{SIS}$, $\mathbf{B}_{(i, j)}=\mathbf{B}_{(j, i)}=0$. Moreover, consider a set $\mathcal{C}=\{4\}$. This is a set of transmitter(s) such that, for a set of receivers $\mathcal{R}_\mathcal{C}=\{1, 2, 4\}$ that can see signals from $\mathcal{C}$ and the set of receivers $\mathcal{R}_{\mathcal{U}_{SIS}}=\{4, 5, 6\}$ that can see signals from $\mathcal{U}_{SIS}$, the following conditions are satisfied: $(i)$ $\mathcal{U}_{SIS}\cap\mathcal{R}_\mathcal{C}=\emptyset$ and  $(ii)$ $\mathcal{R}_{\mathcal{U}_{SIS}}\setminus \mathcal{U}_{SIS} \subseteq \mathcal{R}_\mathcal{C}$. Therefore, by Definition \ref{SecureIndependentSet}, the set $\mathcal{U}_{SIS}=\{5, 6\}$ is a secure independent set.  Similarly, we can show that the sets $\{1\}$, $\{3\}$, $\{2, 6\}$, and $\{4, 6\}$  are also secure independent sets.

\item \textit{(Secure partition)}: From the above secure independent sets and Definition \ref{SecurePartition}, we can create  a secure partition $\mathcal{P}=\{\mathcal{P}_1, \mathcal{P}_2, \dots, \mathcal{P}_5\}$, where $\mathcal{P}_1=\{2, 6\}$, $\mathcal{P}_2=\{1\}$, $\mathcal{P}_3=\{3\}$, $\mathcal{P}_4=\{4\}$, and $\mathcal{P}_5=\{5\}$.
\end{itemize}
As depicted in Fig. \ref{fig:Fig2}\subref{a-SecurePartition1}, we can schedule transmitters in secure partition subsets $\mathcal{P}_1, \mathcal{P}_2, \dots, \mathcal{P}_5$  respectively over $T=|\mathcal{P}|=5$ time slots and have the following sets of transmitters serve as their respective cooperative jamming sets: $\mathcal{C}_1=\emptyset$, $\mathcal{C}_2=\{3\}$, $\mathcal{C}_3=\{2, 5, 6\}$, $\mathcal{C}_4=\{1\}$, and $\mathcal{C}_5=\{4\}$.  Therefore, we are able to separately schedule $|\mathcal{P}|=5$ subsets respectively over  $|\mathcal{P}|=5$ time slots, which leads to a symmetric SDoF of $\frac{1}{|\mathcal{P}|}=\frac{1}{5}$.
\end{manualexample}

\begin{manualexample}{4.A}  \label{4.A} ($5$-user TIM-CM network: Transmission via SP): \normalfont Consider the network depicted in Fig. \ref{fig:Fig3}\subref{a-SecurePartition}. Following a similar analogy as in Example \ref{3.A}, we can use its topology $\mathcal{G}=(\mathcal{T}_1,\mathcal{T}_2, \dots, \mathcal{T}_5,  \mathcal{R}_1, \mathcal{R}_2, \dots, \mathcal{R}_5)$ to obtain three secure partition subsets that collectively contain all five users once. For example, as also shown in Fig. \ref{fig:Fig3}\subref{a-SecurePartition}, we can schedule transmitters in  $\mathcal{P}_1=\{2, 3, 5\}$, $\mathcal{P}_2=\{1\}$, and $\mathcal{P}_3=\{4\}$  respectively over $T=|\mathcal{P}|=|\{\mathcal{P}_1, \mathcal{P}_2, \mathcal{P}_3\}|=3$ time slots and have the following sets serve as their respective cooperative jamming sets: $\mathcal{C}_1=\{1\}$, $\mathcal{C}_2=\{4\}$, and $\mathcal{C}_3=\emptyset$.  Therefore, this leads to a symmetric SDoF of $\frac{1}{|\mathcal{P}|}=\frac{1}{3}$.
\end{manualexample}

Depending on the underlying network topology, it may be possible to achieve more symmetric SDoF than the secure partition scheme of Theorem \ref{TheoremGeneralTIM-CM-SecurePartition}, which allows each user to have only one secure interference free channel use per transmission block.  In particular, each user may be able to use more than one secure interference free channel per transmission block. We will clarify this fact later on using concrete examples. 



The following Theorem states our second lower bound on the symmetric SDoF.

{\color{black}\begin{theorem} \label{TheoremGeneralTIM-CMSDoFNonGreedy}
The symmetric SDoF for the $K$-user TIM-CM network is lower bounded as follows:
\begin{align}
&\textup{\textsf{SDoF}}^{\textup{\textsf{sym}}} \nonumber\\&        \geq\sup_{T\in\mathbb{N}}\max_{\mathcal{U}_{s_1}, \mathcal{U}_{s_2}, \dots, \mathcal{U}_{s_{T}}\in\mathcal{U}^{SIS}}\min_{k\in\{1, 2, \dots, K\}}\hspace{-10pt}\frac{\sum_{t=1}^T\mathbf{1}(k\in{\mathcal{U}}_{s_{t}})}{T},
\end{align}
where $\mathcal{U}^{SIS}=\{\mathcal{U}_{s_1}, \mathcal{U}_{s_2}, \dots, \mathcal{U}_{s_M}\}$ is the set of all secure independent sets for the network topology $\mathcal{G}$. Here, $\mathbf{1}(k\in\mathcal{A})\triangleq 1$, if  $k\in\mathcal{A}$ and $0$ otherwise. 
\end{theorem}}

\begin{proof}
We can take $T$ secure independent sets $\mathcal{U}_{s_1}, \mathcal{U}_{s_2}, \dots, \mathcal{U}_{s_T}\in \mathcal{U}^{SIS}$ and proceed with transmission, while preserving decodability \eqref{DecodabilityConstraint}  and confidentiality \eqref{ConfidentialityConstraint}, as follows. Let all $|\mathcal{U}_{s_t}|$ transmitters  in the set $\mathcal{U}_{s_t}$ send information messages during the same time slot $t\in\{1, 2, \dots, T\}$. Since $\mathcal{U}_{s_t}$ is a secure independent set, there exists a set $\mathcal{C}_t$, $t\in\{1, 2, \dots, T\}$, of cooperative jammers so that all transmitters in  $\mathcal{U}_{s_t}$  can be simultaneously scheduled while satisfying both the decodability and confidentiality constraints.  By doing this for all $T$ sets respectively over $T$ time slots, this implies that each user $k$ will get  a total of $\sum_{t=1}^T\mathbf{1}(k\in\mathcal{U}_{s_t})$ secure interference free channels to use over transmission block $T$. This gives each user $k\in\{1, 2, \dots, K\}$, a total of $\frac{\sum_{t=1}^T\mathbf{1}(k\in\mathcal{U}_{s_t})}{T}$ secure degrees of freedom. Since we are interested in the achievable symmetric SDoF, the largest SDoF that all users can simultaneously achieve for a given choice of $T$ and $\mathcal{U}_{s_1}, \mathcal{U}_{s_2}, \dots, \mathcal{U}_{s_T}$ is given by $\min_{k\in\{1, 2, \dots, K\}} \frac{\sum_{t=1}^T\mathbf{1}(k\in\mathcal{U}_{s_t})}{T}$. Therefore, by optimizing the above over $T$ and $\mathcal{U}_{s_1}, \mathcal{U}_{s_2}, \dots, \mathcal{U}_{s_T}$, we obtain $\sup_{T\in\mathbb{N}}\max_{\mathcal{U}_{s_1}, \mathcal{U}_{s_2}, \dots, \mathcal{U}_{s_T}}\hspace{-2pt}\min_{k\in\{1, 2, \dots, K\}}\hspace{-4pt}\frac{\sum_{t=1}^T\mathbf{1}(k\in\mathcal{U}_{s_t})}{T}$, which is the symmetric SDoF lower bound in Theorem \ref{TheoremGeneralTIM-CMSDoFNonGreedy}.  Clearly, under this transmission scheme, each user may be assigned one or more secure interference free channel uses per transmission blocklength $T$.
\end{proof}

\begin{figure*}
        \captionsetup[subfigure]{aboveskip=-10pt,belowskip=-1pt}
        \centering
        \begin{subfigure}[b]{0.75\textwidth}
                \includegraphics[width=\textwidth]{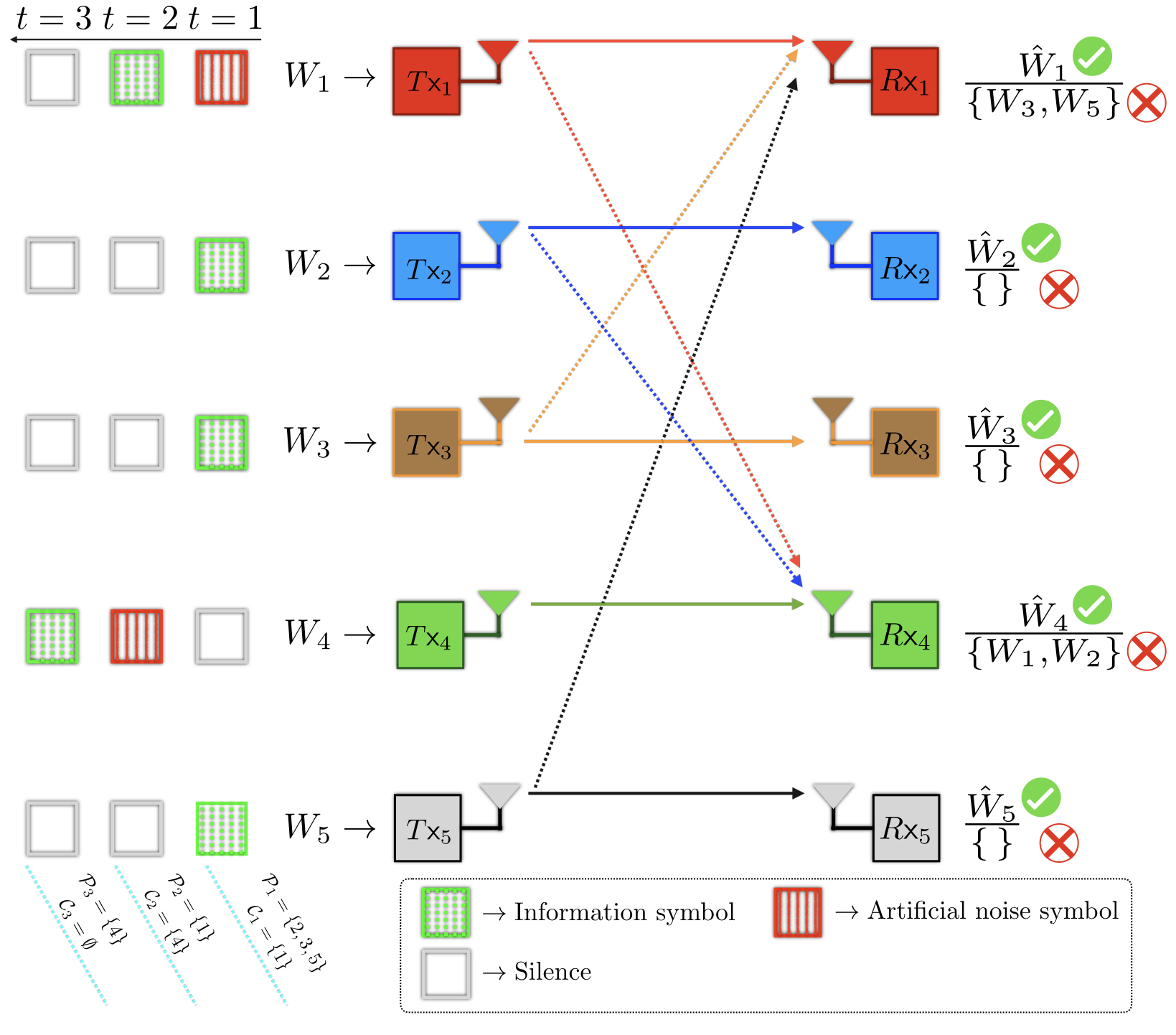}
               \vspace{1pt}
                \caption{}
                \label{a-SecurePartition}
        \end{subfigure}      
        \quad 
        \begin{subfigure}[b]{0.78\textwidth}
                \includegraphics[width=\textwidth]{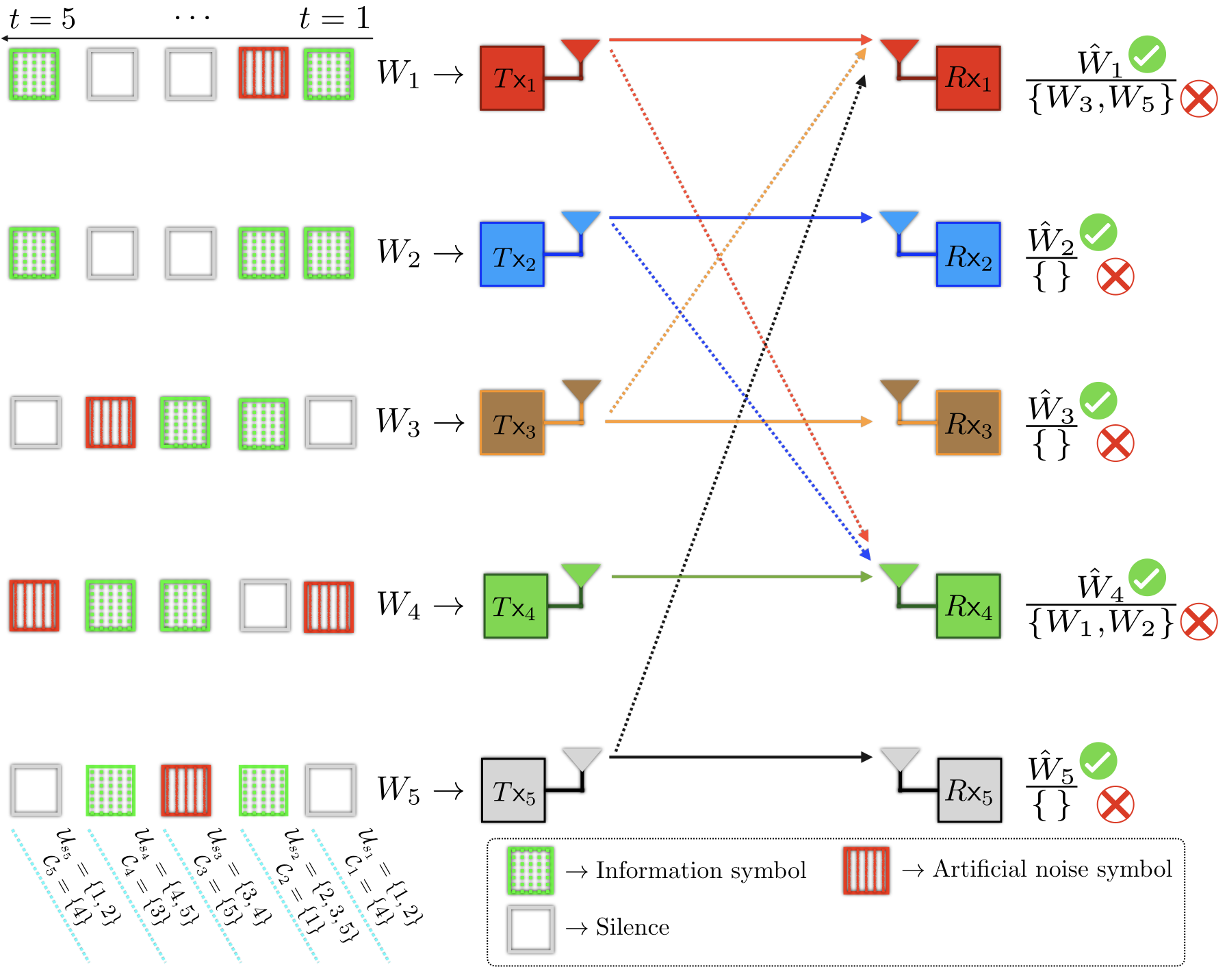}
                \vspace{1pt}
                \caption{}
                \label{b-SecureMaximalIndependentSets}
        \end{subfigure}
        \vspace{5pt}
        \caption{$5$-user TIM-CM network contrasting: (a) Transmission via secure partition  versus  (b) Transmission via secure independent sets.}
        \label{fig:Fig3}
\end{figure*}
We use Examples \ref{3.B} and \ref{4.B} to highlight the principles behind Theorem \ref{TheoremGeneralTIM-CMSDoFNonGreedy}.

\begin{manualexample}{3.B} \label{3.B}($6$-user TIM-CM network: Transmission via SIS): \normalfont Consider the network of Fig. \ref{fig:Fig2}.  From the topology of this network, we can readily infer the following:
\begin{itemize}
\item \textit{(Independent sets)}: According to Definition \ref{IndependentSet}, this topology leads to independent sets $\{3\}$,  $\{4, 6\}$,  $\{1, 5, 6\}$, and $\{2, 5, 6\}$.
\item \textit{(Secure independent sets)}: According to Definition \ref{SecureIndependentSet}, this topology leads to secure independent sets $\{1\}$, $\{3\}$, $\{2, 6\}$,  $\{4, 6\}$, and $\{5, 6\}$.
\end{itemize}

\textcolor{black}{Transmission works as follows: In the first time slot, we can pick any of the three largest secure independent sets $\{2, 6\}$,  $\{4, 6\}$, and $\{5, 6\}$ and schedule its users. Let us, for example, schedule the users in the SIS set $\mathcal{U}_{s_1}=\{2, 6\}$. Note that, as it is also depicted in Fig. \ref{fig:Fig2}\subref{b-SecureMaximalIndependentSets1}, signals sent by Transmitters $2$ and $6$ are only observed at their respectively intended Receivers $2$ and $6$ and not seen at any other receivers as interference. Therefore, the messages from these two transmitter do not need protection at any other receivers. In other words, there is no need for any of the remaining four transmitters in the same network to act as a cooperative jammer during the first time slot. Thus, the cooperative jamming set in the first time slot is the empty set denoted by $\mathcal{C}_1=\emptyset$.}

\textcolor{black}{In total, following the same logic,} we can obtain five (not necessarily unique) secure  independent sets that collectively contain each of the five users at least once, namely $\mathcal{U}_{s_1}=\{2, 6\}$, $\mathcal{U}_{s_2}=\{1\}$, $\mathcal{U}_{s_3}=\{3\}$, $\mathcal{U}_{s_4}=\{4, 6\}$, and $\mathcal{U}_{s_5}=\{5, 6\}$. Accordingly, the following (not necessarily unique) sets of transmitters can serve as cooperative jammers, respectively over five time slots: $\mathcal{C}_1=\emptyset$, $\mathcal{C}_2=\{3\}$, $\mathcal{C}_3=\{2, 5, 6\}$, $\mathcal{C}_4=\{1\}$, and $\mathcal{C}_5=\{4\}$. As depicted in Fig. \ref{fig:Fig2}\subref{b-SecureMaximalIndependentSets1}, by allowing the transmitters in the secure independent sets (cooperative jamming sets)  to respectively send their messages (artificial noise symbols) over $T=5$ time slots, one set per time slot, we can achieve a symmetric SDoF of $\frac{1}{5}$. We note here that this SDoF value is equivalent to the $\frac{1}{5}$ that we achieved in Example \ref{3.A} for the same topology through the secure partition scheme of Theorem \ref{TheoremGeneralTIM-CM-SecurePartition}.
\end{manualexample}

Note that: $(i)$  In Example \ref{4.A}, we transmit using the SP scheme of Theorem \ref{TheoremGeneralTIM-CM-SecurePartition} and, thus, only allow one secure interference free channel use per user over the whole transmission blocklength. $(ii)$  On the other hand (for the same exact topology as in Example \ref{4.A}), in Example \ref{4.B} below, we transmit using the SIS scheme of Theorem \ref{TheoremGeneralTIM-CMSDoFNonGreedy} which allows more than one secure interference free channel uses per user.

\begin{manualexample}{4.B}\label{4.B}($5$-user TIM-CM network: Transmission via SIS): \normalfont Consider the network in  Fig. \ref{fig:Fig3}. 

Following a similar analogy as in Example \ref{3.B}, we can use its topology $\mathcal{G}=(\mathcal{T}_1,\mathcal{T}_2, \dots, \mathcal{T}_5,  \mathcal{R}_1, \mathcal{R}_2, \dots, \mathcal{R}_5)$ to obtain five (not necessarily unique) secure independent sets that collectively contain each of the five users at least once. For example, as also shown in Fig. \ref{fig:Fig3}\subref{b-SecureMaximalIndependentSets}, we can schedule transmitters in  $\mathcal{U}_{s_1}=\{1, 2\}$, $\mathcal{U}_{s_2}=\{2, 3, 5\}$, $\mathcal{U}_{s_3}=\{3, 4\}$, $\mathcal{U}_{s_4}=\{4, 5\}$, and $\mathcal{U}_{s_5}=\{1, 2\}$ respectively over $T=5$ time slots and have the following (not necessarily unique) sets serve as their respective cooperative jamming sets: $\mathcal{C}_1=\{4\}$, $\mathcal{C}_2=\{1\}$, $\mathcal{C}_3=\{5\}$, $\mathcal{C}_4=\{3\}$, and $\mathcal{C}_5=\{4\}$.  Therefore, we are able to allow two secure interference free channels uses for each user over the transmission block $T$, leading to a symmetric SDoF of $\frac{2}{5}$. We note here that this symmetric SDoF value is greater than the $\frac{1}{3}$ that we achieved in Example \ref{4.A} for the same network through the secure partition scheme of Theorem \ref{TheoremGeneralTIM-CM-SecurePartition}.
\end{manualexample}

\begin{remark}  \label{Remark2}(Transmission without secrecy): For the topology in Fig. \ref{fig:Fig3}, the TIM scheme of \cite{NaderializadehAvestimehr} achieves a DoF of $\frac{1}{2} $. To achieve this, in the absence of secrecy constraints, we can schedule the transmitters in the two independent sets $\{1, 2\}$ and $\{3, 4, 5\}$ to respectively transmit over $T=2$ time slots, one set per time slot, thereby achieving a (nonsecure) symmetric DoF of $\frac{1}{2}$. Moreover, we note here that, this DoF matches the upper bound derived by \cite{Jafar} and \cite{MalekiJafarCadambe}.
\end{remark}

\subsection{Outer Bounds on Symmetric SDoF}\label{SectionOuterBoundsonSymmetricSDoF}

In this Section, we present upper bounds on symmetric SDoF for the TIM-CM problem. The Theorem below states our first upper bound on the symmetric SDoF and its proof is provided in Appendix \ref{appendix:Theorem4UpperBound1Proof}. Here, we first provide a few important definitions before stating the Theorem. 

\textcolor{black}{Let $\mathcal{S}_1$ and $\mathcal{S}_2$ be two arbitrary and disjoint subsets of receivers, i.e., such that $\mathcal{S}_1\cap \mathcal{S}_2=\emptyset$  and $\mathcal{S}_1, \mathcal{S}_2 \subseteq \{1, 2, \dots, K\}$. Suppose  $\mathcal{S}_2=\{m_1, m_2, \dots, m_{|\mathcal{S}_2|}\}$ and let $\mathcal{U}(\mathcal{S}_2)=\{\mathcal{T}_{m_q}\}_{q=1}^{|\mathcal{S}_2|}$ be the collection of all index sets of transmitters whose signals are respectively seen at the receivers in $\mathcal{S}_2$.
Furthermore, define $\mu (\mathcal{S}_1, \mathcal{S}_2) \triangleq \min_{\mathcal{V}_2(\mathcal{S}_2), \mathcal{V}_1(\mathcal{S}_1)}|(\mathcal{U}(\mathcal{S}_2) \setminus  \mathcal{V}_2(\mathcal{S}_2)) \setminus \mathcal{V}_1(\mathcal{S}_1)|$. Suppose  $\mathcal{S}_2=\{m_1, m_2, \dots, m_{|\mathcal{S}_2|}\}$, where $\mathcal{V}_2(\mathcal{S}_2)=  \{j_{1}, j_{2}, \dots, j_{|\mathcal{S}_2|}\}$ is an arbitrary set of indexes $j_{q}\in \mathcal{I}_{m_q}$ for $m_q\in \mathcal{S}_2$ and $\mathcal{V}_1(\mathcal{S}_1)= \{\ell_{1}, \ell_{2}, \dots, \ell_{|\mathcal{S}_1|}\}$ is an arbitrary set of indexes $\ell_{p}\in \mathcal{I}_{r_p}$ for $r_p\in \mathcal{S}_1=\{r_1, r_2, \dots, r_{|\mathcal{S}_1|}\}$. Moreover, let $\mathcal{I}_{k}$ be the index set of all signals that are seen as interference at receiver $R\textsf{x}_{k}$ and $\mathcal{T}_k=\{k\}\cup\mathcal{I}_k$.
\begin{theorem} \label{Theorem4UpperBound1}
The symmetric SDoF for the $K$-user TIM-CM network is upper bounded as follows:
\begin{align}
&\textup{\textsf{SDoF}}^{\textup{\textsf{sym}}}\nonumber\\& \leq \min_{\mathcal{S}1, \mathcal{S}_2}\frac{|\mathcal{S}_1|+\mu (\mathcal{S}_1, \mathcal{S}_2)}{|\mathcal{S}_1|+\sum_{p=1}^{|\mathcal{S}_1|}\mathbf{1} (\mathcal{I}_{r_p}\neq \emptyset)+\sum_{q=1}^{|\mathcal{S}_2|}\mathbf{1} (\mathcal{I}_{m_q}\neq \emptyset)}.
\end{align}
\end{theorem}}

As detailed in Appendix \ref{appendix:Theorem4UpperBound1Proof}, the upper bound in Theorem \ref{Theorem4UpperBound1}  is based on the analysis of the received signal structures at all $K$ receivers and their respective interference components. To prove the Theorem, we first show how to get a nontrivial upper bound on secrecy rate by discarding all but one interference components from the received signal. Next,  we show that the secrecy rate of the resulting signal can be bounded in two different ways, using which we then deduce the desired upper bound. We use Examples \ref{3.C} and \ref{4.C} to highlight the principles behind  Theorem \ref{Theorem4UpperBound1}.

\begin{manualexample}{3.C}\label{3.C}($6$-user TIM-CM network: Upper bound on symmetric SDoF): \normalfont
Let us revisit the $6$-user TIM-CM network in Fig. \ref{fig:Fig2}.  Our aim is to show that we can obtain a symmetric SDoF upper bound of $\frac{1}{3}$.  To this end, we can now proceed as follows: Let us pick two disjoint subsets of receivers $\mathcal{S}_1=\{2, 4, 6\}$ and $\mathcal{S}_2=\{1, 3, 5\}$.  We note here that the users in $\mathcal{S}_1$ have interference sets  $\mathcal{I}_2=\{1, 3, 4\}$, $\mathcal{I}_4=\{3, 5\}$ and $\mathcal{I}_6=\{3\}$, whereas the users in $\mathcal{S}_2$ have interference sets $\mathcal{I}_1=\{4\}$, $\mathcal{I}_3=\{1\}$ and $\mathcal{I}_5=\{3\}$, respectively. Therefore, $\sum_{p=1}^{|\mathcal{S}_1|}\mathbf{1} (\mathcal{I}_{r_p}\neq \emptyset)=\sum_{r_p\in\mathcal{S}_1=\{2, 4, 6\}}\mathbf{1} (\mathcal{I}_{r_p}\neq \emptyset)=3$ and $\sum_{q=1}^{|\mathcal{S}_2|}\mathbf{1} (\mathcal{I}_{m_q}\neq \emptyset)=\sum_{m_q\in\mathcal{S}_2=\{1, 3, 5\}}\mathbf{1} (\mathcal{I}_{m_q}\neq \emptyset)=3$. Moreover, we can, for example, pick the multi-set $\mathcal{U}(\mathcal{S}_2)$  and the two sets $\mathcal{V}_2(\mathcal{S}_2)$ and $\mathcal{V}_1(\mathcal{S}_1)$ as follows:
\begin{itemize}
\item $\mathcal{U}(\mathcal{S}_2)=\{\mathcal{T}_1, \mathcal{T}_3, \mathcal{T}_5\}=\{\{1, 4\}, \{1, 3\},\{3, 5\}\}=\{1, 1, 3, 3, 4, 5\}$.
\item  $\mathcal{V}_2(\mathcal{S}_2)=\{4, 1, 3\}$, where $4\in\mathcal{I}_1$, $1\in\mathcal{I}_3$, and $3\in\mathcal{I}_5$.  
\item  $\mathcal{V}_1(\mathcal{S}_1)=\{1, 5, 3\}$, where $1\in\mathcal{I}_2$, $5\in\mathcal{I}_4$, and $3\in\mathcal{I}_6$.  
\end{itemize}
Thus, using its definition, we can calculate $\mu (\mathcal{S}_1, \mathcal{S}_2)$ as follows:
\begin{align*}
\mu (\mathcal{S}_1, \mathcal{S}_2)&= |(\mathcal{U}(\mathcal{S}_2) \setminus  \mathcal{V}_2(\mathcal{S}_2)) \setminus \mathcal{V}_1(\mathcal{S}_1)| \\&\hspace{-15pt}= |(\{1, 1, 3, 3, 4, 5\}\setminus \{4, 1, 3\})\setminus \{1, 5, 3\}|
= 0.
\end{align*}
Therefore, using the above, we obtain the following upper bound on symmetric SDoF:
\begin{align*}
     &\textup{\textsf{SDoF}}^{\textup{\textsf{sym}}}\\ &\leq  \min_{\mathcal{S}1, \mathcal{S}_2}\frac{|\mathcal{S}_1|+\mu (\mathcal{S}_1, \mathcal{S}_2)}{|\mathcal{S}_1|+\sum_{p=1}^{|\mathcal{S}_1|}\mathbf{1} (\mathcal{I}_{r_p}\neq \emptyset)+\sum_{q=1}^{|\mathcal{S}_2|}\mathbf{1} (\mathcal{I}_{m_q}\neq \emptyset)}
    \\&\leq \frac{1}{3}.
\end{align*}
\end{manualexample}

\begin{manualexample}{4.C}\label{4.C}($5$-user TIM-CM network: Upper bound on symmetric SDoF): \normalfont
Let us revisit the $5$-user TIM-CM network in Fig. \ref{fig:Fig3}.  Our aim is to show that we can obtain a symmetric SDoF upper bound of $\frac{1}{2}$, which is thus close to the achievable lower bound of $\frac{2}{5}$ that we derived in Example \ref{4.B} for the same topology (via secure independent sets).  To this end, we can now proceed as follows: Let us pick two disjoint subsets of receivers $\mathcal{S}_1=\{1, 4\}$ and $\mathcal{S}_2=\{2, 3\}$.  We note here that the users in $\mathcal{S}_1$ have interference sets $\mathcal{I}_1=\{3, 5\}$ and $\mathcal{I}_4=\{1, 2\}$, whereas the users in $\mathcal{S}_2$ have interference sets $\mathcal{I}_2=\emptyset$ and $\mathcal{I}_3=\emptyset$, respectively. Therefore, $\sum_{p=1}^{|\mathcal{S}_1|}\mathbf{1} (\mathcal{I}_{r_p}\neq \emptyset)=\sum_{r_p\in\mathcal{S}_1=\{1, 4\}}\mathbf{1} (\mathcal{I}_{r_p}\neq \emptyset)=2$ and $\sum_{q=1}^{|\mathcal{S}_2|}\mathbf{1} (\mathcal{I}_{m_q}\neq \emptyset)=\sum_{m_q\in\mathcal{S}_2=\{2, 3\}}\mathbf{1} (\mathcal{I}_{m_q}\neq \emptyset)=0$. Moreover, we can, for example, pick the multi-set $\mathcal{U}(\mathcal{S}_2)$  and the two sets $\mathcal{V}_2(\mathcal{S}_2)$ and $\mathcal{V}_1(\mathcal{S}_1)$ as follows:
\begin{itemize}
\item $\mathcal{U}(\mathcal{S}_2)=\{\mathcal{T}_2, \mathcal{T}_3\}=\{\{2\}, \{3\}\}=\{2, 3\}$.
\item  $\mathcal{V}_2(\mathcal{S}_2)=\{3, 2\}$, where $3\in\mathcal{I}_1$ and $2\in\mathcal{I}_4$. 
\item  $\mathcal{V}_1(\mathcal{S}_1)=\emptyset$, because  $\mathcal{I}_2=\emptyset$ and $\mathcal{I}_3=\emptyset$. 
\end{itemize}
Thus, using its definition, we can calculate $\mu (\mathcal{S}_1, \mathcal{S}_2)$ as follows:
\begin{align*}
\mu (\mathcal{S}_1, \mathcal{S}_2)&= |(\mathcal{U}(\mathcal{S}_2) \setminus  \mathcal{V}_2(\mathcal{S}_2)) \setminus \mathcal{V}_1(\mathcal{S}_1)|
\\&= |(\{2, 3\}\setminus \{3, 2\})\setminus \emptyset| 
= 0.
\end{align*}
Therefore, using the above, we obtain the following upper bound on symmetric SDoF:
\begin{align*}
     &\textup{\textsf{SDoF}}^{\textup{\textsf{sym}}} \\&\leq  \min_{\mathcal{S}1, \mathcal{S}_2}\frac{|\mathcal{S}_1|+\mu (\mathcal{S}_1, \mathcal{S}_2)}{|\mathcal{S}_1|+\sum_{p=1}^{|\mathcal{S}_1|}\mathbf{1} (\mathcal{I}_{r_p}\neq \emptyset)+\sum_{q=1}^{|\mathcal{S}_2|}\mathbf{1} (\mathcal{I}_{m_q}\neq \emptyset)}
    \\&\leq \frac{1}{2}.
\end{align*}
\end{manualexample}

We next present our second upper bound in Theorem \ref{Theorem5UpperBound2}, which we later show to be tighter for some examples compared to the results of Theorem \ref{Theorem4UpperBound1}. The main difference between this upper bound and the first one is that  we now exploit the potential presence of \textit{fractional signal generators} within the network, a concept that was introduced in  \cite{NaderializadehAvestimehr, Gesbert} for the nonsecrecy constrained TIM problem. This is a paradigm where the interference signal component of the received signal at some receiver can generate either a statistically equivalent version of the received signal at some other receiver or its cleaner version, i.e., with less interference. To this end, let us first define a fractional signal generator. 

\begin{definition}\label{FractionalGenerators} (Fractional signal generator):  \textcolor{black}{Let $\mathcal{I}_{k}$ be the index set of all signals that are seen as interference at receiver $R\textsf{x}_{k}$ and $\mathcal{T}_k=\{k\}\cup\mathcal{I}_k$}. We say that a receiver $R\textsf{x}_k$ is a fractional signal generator of  $\mathcal{G}_k$, if $\mathcal{G}_k\subseteq \mathcal{I}_k$ and there exists a permutation $\Pi^{(k)}=(\Pi_1, \Pi_2, \dots \Pi_{|\mathcal{G}_k|})$ of the set $\mathcal{G}_k$ such that $\mathcal{G}_k\setminus\{\Pi_1, \Pi_2, \dots, \Pi_i\}\subseteq \mathcal{I}_{\Pi_i}$, for $i=1, 2, \dots, |\mathcal{G}_k|$.
\end{definition}

Consider the $6$-user network topology depicted in Fig. \ref{fig:Fig2}. \textcolor{black}{In this network, $R\textsf{x}_2$ receives signals from Transmitters $1$, $2$, $3$, and $4$. This fact can be equivalently denoted as $\mathcal{T}_2=\{2\}\cup\mathcal{I}_2=\{1, 2, 3, 4\}$. The main idea behind Definition \ref{FractionalGenerators} is that, by following a not necessarily unique successive order, $R\textsf{x}_2$ can generate the signal index set $\mathcal{G}_2=\{3, 4\}$ because $\mathcal{G}_2\subseteq\mathcal{I}_2=\{1, 3, 4\}$. More specifically, we next demonstrate why $R\textsf{x}_2$ is called a fractional signal generator for the signal index set $\mathcal{G}_2$:}
\begin{itemize}
\item \textcolor{black}{First, we show how $R\textsf{x}_2$ can generate the second element of $\mathcal{G}_2$, i.e., $4$. Because the order of signal generation matters, we denote this fact by $(\{3, 4\}\setminus \{4\})=\{3\}\subseteq \mathcal{I}_4=\{3, 5\}$. This indicates that by first generating $4$, we only remain with $\{3\}$ which is a subset of the interference set $\mathcal{I}_4=\{3,5\}$ at $R\textsf{x}_4$. As we will show in the second step below, this order of signal generation will enable the generation of the first element of $\mathcal{G}_2$, i.e., $3$. To this end, assume that the interference component $X^{n}_1$ seen from $T\textsf{x}_1$ is provided by a genie to $R\textsf{x}_2$. In turn, this allows $R\textsf{x}_2$ to discard this interference signal component from its composite interference component in order to generate a signal $\tilde{Y}^{n}_4$ with components $\{3, 4\}$, which is an enhanced version of the original signal $Y^{n}_4$, whose components are $\mathcal{T}_4=\{3, 4, 5\}$. In other words, the newly generated signal $\tilde{Y}^{n}_4$ is the enhanced version $Y^{n}_4$ because $\tilde{Y}^{n}_4$ has only one interference term, i.e., sees signals from $T\textsf{x}_3$, whereas  the original signal $Y^{n}_4$ has two interference terms, i.e., sees signals from $T\textsf{x}_3$ and $T\textsf{x}_5$. Therefore, if $R\textsf{x}_4$ can decode its intended message $W_4$ from the original signal $Y^{n}_4$, then its enhanced version $\tilde{R}\textsf{x}_4$ can also decode $W_4$ from  $\tilde{Y}^{n}_4$.}

\item \textcolor{black}{Second, we show how further processing of the signal at $R\textsf{x}_2$ can generate the first element of the $\mathcal{G}_2$, i.e., $3$. We denote this by $(\{3\}\setminus \{3\})=\emptyset$ to indicate that $3$ is the element of $\mathcal{G}_2$ which is generated the last. The intuition behind this second step is as follows: From the interference subset $\{3, 4\}$ at $R\textsf{x}_2$ (i.e., having already removed the interference term from  $T\textsf{x}_1$ in the above step), one can next discard the (genie provided) interference component seen from $T\textsf{x}_4$ in order to generate a signal $\tilde{Y}^{n}_3$ with  a single component $\{3\}$, which is the enhanced version of the original signal $Y^{n}_3$, whose components are $\mathcal{T}_3=\{1, 3\}$. Therefore, if $R\textsf{x}_3$ can decode its intended message $W_3$ from the original signal $Y^{n}_3$, then $\tilde{R}\textsf{x}_3$ can also decode $W_3$ from  $\tilde{Y}^{n}_3$.}
\end{itemize}

Therefore, successive decodability of $W_3$ and $W_4$, from the signals generated by the interference signal at $R\textsf{x}_2$, follows the sequence $\Pi^{(2)}=(4, 3)$, which satisfies Definition \ref{FractionalGenerators}. 

The Theorem below states our second upper bound on the symmetric SDoF and its proof is provided in Appendix \ref{appendix:Theorem5UpperBound2Proof}. Here, we first provide a few important definitions before stating the Theorem. 

\textcolor{black}{Let $\mathcal{S}_1$ and $\mathcal{S}_2$ be two arbitrary and disjoint subsets of receivers, i.e., such that $\mathcal{S}_1\cap \mathcal{S}_2=\emptyset$  and $\mathcal{S}_1, \mathcal{S}_2 \subseteq \{1, 2, \dots, K\}$. Suppose  $\mathcal{S}_2=\{m_1, m_2, \dots, m_{|\mathcal{S}_2|}\}$ and let  $\mathcal{U}(\mathcal{S}_2)=\{\mathcal{T}_{m_q}\}_{q=1}^{|\mathcal{S}_2|}$ be  the collection of all index sets of transmitters whose signals are respectively seen at the receivers in $\mathcal{S}_2$. Consider a collection of sets $\mathcal{G}^{\mathcal{S}_2}=\{\mathcal{G}^{(2)}_{m_1}, \mathcal{G}^{(2)}_{m_2}, \dots, \mathcal{G}^{(2)}_{m_{|\mathcal{S}_2|}}\}$ such that Receiver $m_q$ is a generator of $\mathcal{G}^{(2)}_{m_q}\subseteq \mathcal{I}_{m_q}$ for all $q\in\{1, 2, \dots, |\mathcal{S}_2|\}$, with a permutation $\Pi^{(m_q)}=\{\Pi^{(m_q)}_1, \Pi^{(m_q)}_2, \dots \Pi^{(m_q)}_{|\mathcal{G}^{(2)}_{m_q}|}\}$ according to Definition \ref{FractionalGenerators}. Then, let $\mathcal{\tilde{V}}(\mathcal{S}_2)= \{\Pi^{(m_1)}_{|\mathcal{G}^{(2)}_{m_1}|}, \Pi^{(m_2)}_{|\mathcal{G}^{(2)}_{m_2}|}, \dots \Pi^{(m_{|\mathcal{S}_2|})}_{|\mathcal{G}^{(2)}_{m_{|\mathcal{S}_2|}}|}\}$ be an arbitrary set of indexes $ \Pi^{(m_q)}_{|\mathcal{G}^{(2)}_{m_q}|}\in \mathcal{I}_{m_q}$ for $m_q\in \mathcal{S}_2$. Similarly, suppose $\mathcal{S}_1=\{r_1, r_2, \dots r_{|\mathcal{S}_1|}\}$ and consider a collection of sets $\mathcal{G}^{\mathcal{S}_1}=\{\mathcal{G}^{(1)}_{r_1}, \mathcal{G}^{(1)}_{r_2}, \dots, \mathcal{G}^{(1)}_{r_{|\mathcal{S}_1|}}\}$ such that Receiver $r_p$ is a generator of $\mathcal{G}^{(1)}_{r_p}\subseteq \mathcal{I}_{r_p}$ for all $p\in\{1, 2, \dots, |\mathcal{S}_1|\}$, with a permutation $\Pi^{(r_p)}=\{\Pi^{(r_p)}_1, \Pi^{(r_p)}_2, \dots \Pi^{(r_p)}_{|\mathcal{G}^{(1)}_{r_p}|}\}$ according to Definition \ref{FractionalGenerators}. Then, let $\mathcal{\tilde{V}}_1(\mathcal{S}_1))= \{\Pi^{(r_1)}_{|\mathcal{G}^{(1)}_{r_1}|}, \Pi^{(r_2)}_{|\mathcal{G}^{(1)}_{r_2}|}, \dots \Pi^{(r_{|\mathcal{S}_1|})}_{|\mathcal{G}^{(1)}_{r_{|\mathcal{S}_1|}}|}\}$ be an arbitrary set of indexes $ \Pi^{(r_p)}_{|\mathcal{G}^{(1)}_{r_p}|}\in \mathcal{I}_{r_p}$ for $r_p\in \mathcal{S}_1$. Define $\tilde{\mu} (\mathcal{S}_1, \mathcal{S}_2) \triangleq |(\mathcal{U}(\mathcal{S}_2) \setminus  \mathcal{\tilde{V}}_2(\mathcal{S}_2)) \setminus \mathcal{\tilde{V}}_1(\mathcal{S}_1))|$ and let $\mathcal{I}_{k}$ be the index set of all signals that are seen as interference at receiver $R\textsf{x}_{k}$ and $\mathcal{T}_k=\{k\}\cup\mathcal{I}_k$.
\begin{theorem} \label{Theorem5UpperBound2}
The symmetric SDoF for the $K$-user TIM-CM network is upper bounded as follows:
\begin{align}
&\textup{\textsf{SDoF}}^{\textup{\textsf{sym}}}\nonumber\\&  \leq \min_{(\mathcal{S}1, \mathcal{S}_2)} \min_{(\mathcal{G}^{\mathcal{S}_1}, \mathcal{G}^{\mathcal{S}_2})}\frac{|\mathcal{S}_1|+\tilde{\mu} (\mathcal{S}_1, \mathcal{S}_2)}{|\mathcal{S}_1|+\sum_{p=1}^{|\mathcal{S}_1|}|\mathcal{G}^{(1)}_{r_p}|+\sum_{q=1}^{|\mathcal{S}_2|}|\mathcal{G}^{(2)}_{m_q}|}.
\end{align}
\end{theorem}}

As detailed in Appendix \ref{appendix:Theorem5UpperBound2Proof}, the upper bound in Theorem \ref{Theorem5UpperBound2} is also based on the analysis of the received signal structures at all $K$ receivers and their respective interference components with regard to the underlying arbitrary topology.  To prove the Theorem, we first show how to get a nontrivial bound on secrecy rate by discarding all the interference components in the received signal,  except for those that can be sequentially generated from the interference signal according to Definition \ref{FractionalGenerators}. Next, we show that the secrecy rate of the resulting signal can be bounded using two types of bounds from which we are then able to deduce the desired upper bound. We use Example \ref{3.D} to highlight the principles behind Theorem \ref{Theorem5UpperBound2}.

\begin{manualexample}{3.D}\label{3.D} ($6$-user TIM-CM network:  Upper bound on symmetric SDoF): \normalfont
Let us revisit the $6$-user TIM-CM network in Fig. \ref{fig:Fig2}.  Our aim is to show that we can obtain a symmetric SDoF upper bound of $\frac{1}{4}$ which is a tighter bound compared to the $\frac{1}{3}$ obtained in Example \ref{3.C}. We note here that  $\frac{1}{4}$ is close to the achievable lower bound of $\frac{1}{5}$ that we derived in Example \ref{3.A}  (via secure partition) and \ref{3.B} (via secure sets) for the same network topology. 

We can now proceed as follows: Let us, for example, pick two disjoint subsets receivers $\mathcal{S}_1=\{2\}$ and $\mathcal{S}_2=\{3\}$.  According to Definition \ref{FractionalGenerators} of fractional signal generators, Receiver $2\in\mathcal{S}_1$ (which has the interference set $\mathcal{I}_2=\{1, 3, 4\}$), can generate the following (not necessarily unique) set $\mathcal{G}^{(1)}_2=\{3, 4\}$, with a permutation sequence $\Pi^{(2)}=(4, 3)$. Therefore, this leads to the multi-set $\mathcal{G}^{\mathcal{S}_1}=\mathcal{G}^{(1)}_2=\{3, 4\}$, which is the collection of all signal sets generated by receivers in $\mathcal{S}_1$.  Following a similar analogy, Receiver $3\in\mathcal{S}_2$ (which has the interference set $\mathcal{I}_3=\{1\}$) can only generate $\mathcal{G}^{(2)}_3=\{1\}$. Therefore, $\mathcal{G}^{\mathcal{S}_2}=\mathcal{G}^{(2)}_3=\{1\}$.  This implies that $\sum_{p=1}^{|\mathcal{S}_1|}|\mathcal{G}^{(1)}_{r_p}|=\sum_{r_p\in\mathcal{S}_1=\{2\}}|\mathcal{G}^{(1)}_{r_p}|=2$, whereas  $\sum_{q=1}^{|\mathcal{S}_2|}|\mathcal{G}^{(2)}_{m_q}|=\sum_{m_q\in\mathcal{S}_2=\{3\}}|\mathcal{G}^{(2)}_{m_q}|=1$. Moreover, we can, for example, pick the multi-set $\mathcal{U}(\mathcal{S}_2)$  and the two sets $\mathcal{\tilde{V}}(\mathcal{S}_2)$ and $\mathcal{\tilde{V}}_1(\mathcal{S}_1)$ as follows:
\begin{itemize}
\item $\mathcal{U}(\mathcal{S}_2)=\{\mathcal{T}_3\}=\{1, 3\}$.
\item  $\mathcal{\tilde{V}}(\mathcal{S}_2)=\{\Pi^{(m_{|\mathcal{S}_2|})}_{|\mathcal{G}^{(2)}_{m_{|\mathcal{S}_2|}}|}\}=\{\Pi^{(3)}_{|\mathcal{G}^{(2)}_3|}\}=\{1\}$, where $1\in\mathcal{G}^{(2)}_3\subseteq\mathcal{I}_3$. 
\item  $\mathcal{\tilde{V}}_1(\mathcal{S}_1))=\{\Pi^{(r_{|\mathcal{S}_1|})}_{|\mathcal{G}^{(1)}_{r_{|\mathcal{S}_1|}}|}\}=\{\Pi^{(2)}_{|\mathcal{G}^{(1)}_2|}\}=\{3\}$, where $3\in\mathcal{G}^{(1)}_2\subseteq\mathcal{I}_2$. 
\end{itemize}
Thus, using its definition, we can calculate $\tilde{\mu} (\mathcal{S}_1, \mathcal{S}_2)$ as follows:
\begin{align*}
\tilde{\mu} (\mathcal{S}_1, \mathcal{S}_2)&=|(\mathcal{U}(\mathcal{S}_2) \setminus  \mathcal{\tilde{V}}_2(\mathcal{S}_2)) \setminus \mathcal{\tilde{V}}_1(\mathcal{S}_1))|
\\&=|(\{1, 3\}\setminus \{1\})\setminus \{3\}|
= 0.
\end{align*}
Therefore, using the above, we obtain the following upper bound on symmetric SDoF:
\begin{align*}
     &\textup{\textsf{SDoF}}^{\textup{\textsf{sym}}} \\&\leq  \min_{(\mathcal{S}1, \mathcal{S}_2)} \min_{(\mathcal{G}^{\mathcal{S}_1}, \mathcal{G}^{\mathcal{S}_2})}\frac{|\mathcal{S}_1|+\tilde{\mu} (\mathcal{S}_1, \mathcal{S}_2)}{|\mathcal{S}_1|+\sum_{p=1}^{|\mathcal{S}_1|}|\mathcal{G}^{(1)}_{r_p}|+\sum_{q=1}^{|\mathcal{S}_2|}|\mathcal{G}^{(2)}_{m_q}|}
     \\&\leq \frac{1}{4}.\nonumber
\end{align*}
\end{manualexample}

We now remark the following regarding the $5$-user example network topology of Fig. \ref{fig:Fig3}.
\begin{remark} \label{RemarkAboutFig3} Since the topology of Fig. \ref{fig:Fig3} that we used in Example \ref{4.C} does not contain nontrivial fractional signal generators (i.e., where some of the receivers can generate more than one signal), applying the principles of Theorem \ref{Theorem5UpperBound2} would also lead to the symmetric SDoF upper bound of $\frac{1}{2}$.
\end{remark}

\subsection{Case Studies: Settling Symmetric SDoF for $2$-user and $3$-user Topologies}\label{SectionCaseStudies}
\begin{figure*}[!t]
	\begin{center}
		\includegraphics[width=0.7\textwidth]{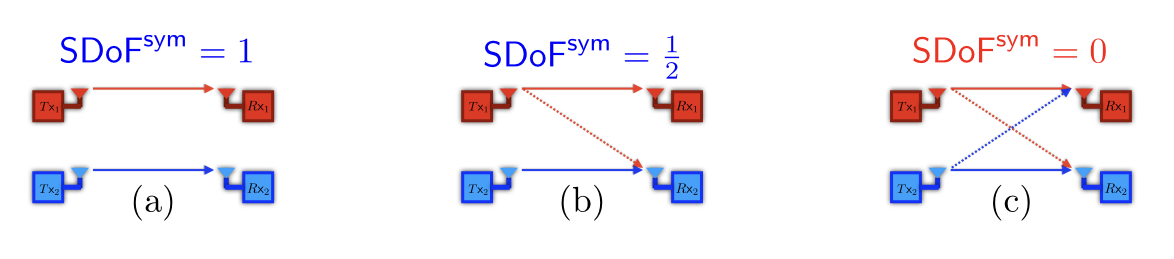}
	\vspace{0pt}
		\caption{Case study A: Application of current results to all $2$-user TIM-CM networks.}
		\label{fig:Fig5}
	\end{center}
	\vspace{0pt}
\end{figure*} 

\begin{figure*}[!t]
	\begin{center}
		\includegraphics[width=0.7\textwidth]{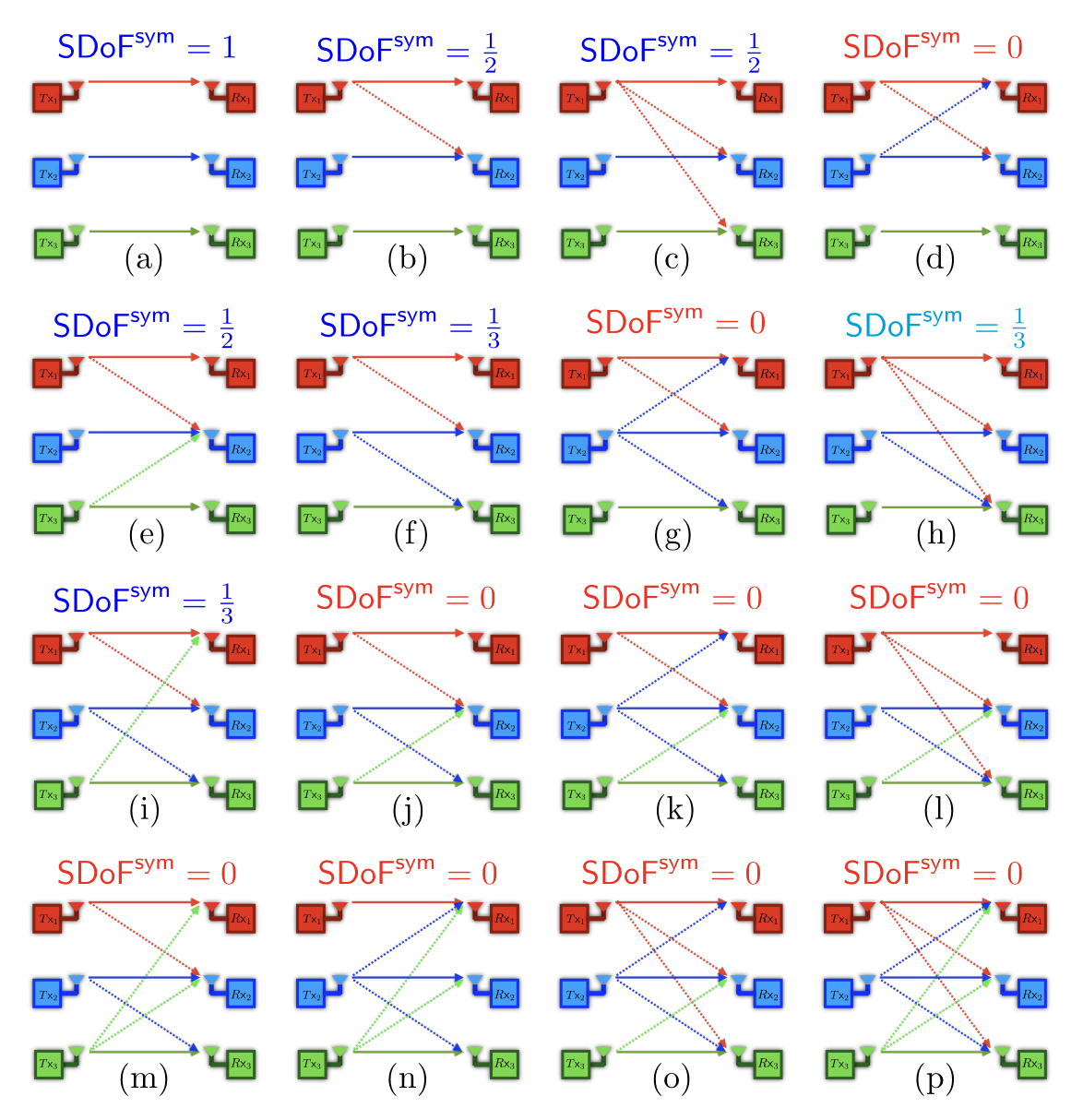}
	\vspace{0pt}
		\caption{Case study B: Application of current results to all $3$-user TIM-CM networks.}
		\label{fig:Fig6}
	\end{center}
	\vspace{0pt}
\end{figure*} 

In this Section, we present two case studies: In the first one, we study the $2$-user TIM-CM networks and, in the second one, we study the $3$-user TIM-CM networks.  To this end, we enumerate all possible (non-isomorphic) topologies for the $2$-user and $3$-user TIM-CM networks as respectively depicted in Fig. \ref{fig:Fig5} and Fig. \ref{fig:Fig6}. We then apply the result of Theorem \ref{TheoremSDoFSymmetry} to eliminate all topologies where positive symmetric SDoF is not feasible.  To find the outer bounds for the remaining networks, i.e., where positive symmetric SDoF is feasible, we apply the upper bound results of Theorems  \ref{Theorem4UpperBound1} and \ref{Theorem5UpperBound2}. We then verify the optimality of symmetric SDoF by showing that these upper bounds match the lower bounds by applying the achievable result of Theorem \ref{TheoremGeneralTIM-CM-SecurePartition} or \ref{TheoremGeneralTIM-CMSDoFNonGreedy}, thereby settling the $K$-user TIM-CM problem for $K\leq 3$.

Our results can be readily applied to calculate the exact symmetric SDoF for all (nonisomorphic)  $2$-user TIM-CM topologies as respectively given in Fig. \ref{fig:Fig5} (a)-(c).  To avoid repetitiveness, here, we only show how to find the exact symmetric SDoF for all $3$-user TIM-CM topologies. To this end, we first consider all the $3$-user TIM-CM networks for which symmetric SDoF is zero, i.e., the topologies that satisfy the conditions of Theorem \ref{TheoremSDoFSymmetry}. As indicated in Fig. \ref{fig:Fig6}, symmetric SDoF is zero for networks  in Fig. \ref{fig:Fig6} (d), (g), (j), (k), (l), (m), (n), (o), and (p). As an illustrative example, we provide calculation of the exact symmetric SDoF for the network of Fig. \ref{fig:Fig6} (o). The calculations for the other topologies (for which symmetric SDoF is zero) follow a similar analogy and are thus omitted here to avoid repetition. 

Consider the network of Fig. \ref{fig:Fig6} (o). At receiver $R\textsf{x}_2$, we have $\mathcal{I}_2=\{1, 3\}$ and $\mathcal{T}_2=\{2\}\cup \mathcal{I}_2=\{1, 2, 3\}$. At receiver $R\textsf{x}_3$, we have $\mathcal{I}_3=\{1, 2\}$ and $\mathcal{T}_3=\{3\}\cup \mathcal{I}_3=\{1, 2, 3\}$. We note here that for the two transmit-receive pairs $T\textsf{x}_2-R\textsf{x}_2$ and $T\textsf{x}_3-R\textsf{x}_3$, the conditions of Theorem \ref{TheoremSDoFSymmetry} are satisfied because $(i)$ $2\in\mathcal{I}_3$, $(ii)$ $3\in\mathcal{I}_2$, and $(iii)$ $\mathcal{I}_3\setminus\{2\}=\mathcal{I}_2\setminus\{3\}$. Since the conditions of Theorem \ref{TheoremSDoFSymmetry} are satisfied, then symmetric SDoF is zero for this network.

Second, we consider all the $3$-user TIM-CM networks for which SDoF is nonzero, i.e., the topologies that do not satisfy the conditions of Theorem \ref{TheoremSDoFSymmetry}. As indicated in Fig. \ref{fig:Fig6}, symmetric SDoF is nonzero for networks  in Fig. \ref{fig:Fig6} (a), (b), (c), (e), (f), (h), and (i). It can easily be verified that for each of these topologies, except for the topology of Fig. \ref{fig:Fig6} (h), the lower bound achieved using the secure partition scheme of Theorem \ref{TheoremGeneralTIM-CM-SecurePartition} matches the upper bound obtained through the result of Theorem \ref{Theorem4UpperBound1}, thereby leading to their exact symmetric SDoF. For the network of Fig. \ref{fig:Fig6} (h), we apply the result of Theorem \ref{TheoremGeneralTIM-CM-SecurePartition} to find the lower bound. However, we use the result of Theorem \ref{Theorem5UpperBound2} to find its matching upper bound instead of Theorem \ref{Theorem4UpperBound1} because the latter leads to a loose bound.

As illustrative examples, we provide calculations of exact symmetric SDoF for the networks of Fig. \ref{fig:Fig6} (e) and (h). The calculations for the other topologies (for which symmetric SDoF is nonzero) follow a similar analogy and are thus omitted here to avoid repetition. 

Consider the $3$-user network of Fig. \ref{fig:Fig6} (e). At receiver $R\textsf{x}_1$, we have $\mathcal{I}_1=\emptyset$ and $\mathcal{T}_1=\{1\}\cup \mathcal{I}_1=\{1\}$.  At $R\textsf{x}_2$, we have $\mathcal{I}_2=\{1, 3\}$ and $\mathcal{T}_2=\{2\}\cup \mathcal{I}_2=\{1, 2, 3 \}$.  At $R\textsf{x}_3$, we have $\mathcal{I}_3=\emptyset$ and $\mathcal{T}_3=\{3\}\cup \mathcal{I}_3=\{3\}$.  Then, we can apply the scheme of Theorem \ref{TheoremGeneralTIM-CM-SecurePartition} by creating a secure partition $\mathcal{P}=\{\mathcal{P}_1=\{1, 3\}, \mathcal{P}_2=\{2\}\}$ and scheduling the transmitters in  $\mathcal{P}_1$ to  send their respective messages over the first time slot (while the second transmitter $T\textsf{x}_2$ acts as a cooperative jammer by sending an artificial noise symbol over the first time slot, i.e., $\mathcal{C}_1=\{2\}$). We schedule the transmitter in  $\mathcal{P}_2$ to send its message over the second time slot (while $T\textsf{x}_1$ and $T\textsf{x}_3$  remain silent over the second time slot, i.e., $\mathcal{C}_2=\emptyset$). Therefore, we achieve a symmetric SDoF of $\frac{1}{|\mathcal{P}|}=\frac{1}{2}$ per user.

To find the upper bound on symmetric SDoF for  the network of Fig. \ref{fig:Fig6} (e), we apply the result of Theorem \ref{Theorem4UpperBound1} as as shown next. Let us pick two disjoint subsets of receivers $\mathcal{S}_1=\{2\}$ and $\mathcal{S}_2=\{3\}$.  We note here that the user in $\mathcal{S}_1$ has interference set $\mathcal{I}_2=\{1, 3\}$, whereas the user in $\mathcal{S}_2$ has interference set $\mathcal{I}_3=\emptyset$. Therefore, $\sum_{p=1}^{|\mathcal{S}_1|}\mathbf{1} (\mathcal{I}_{r_p}\neq \emptyset)=\sum_{r_p\in\mathcal{S}_1=\{2\}}\mathbf{1} (\mathcal{I}_{r_p}\neq \emptyset)=1$ and $\sum_{q=1}^{|\mathcal{S}_2|}\mathbf{1} (\mathcal{I}_{m_q}\neq \emptyset)=\sum_{m_q\in\mathcal{S}_2=\{3\}}\mathbf{1} (\mathcal{I}_{m_q}\neq \emptyset)=0$. Moreover, we can, for example, pick the multi-set $\mathcal{U}(\mathcal{S}_2)$  and the two sets $\mathcal{V}_2(\mathcal{S}_2)$ and $\mathcal{V}_1(\mathcal{S}_1)$ as follows: $\mathcal{U}(\mathcal{S}_2)=\{\mathcal{T}_3\}=\{3\}$. $\mathcal{V}_2(\mathcal{S}_2)=\emptyset$ because  $\mathcal{I}_3=\emptyset$. $\mathcal{V}_1(\mathcal{S}_1)=\{3\}$, where   $3\in\mathcal{I}_2=\{1, 3\}$. Thus, using its definition, we can calculate $\mu (\mathcal{S}_1, \mathcal{S}_2)$ as follows: $\mu (\mathcal{S}_1, \mathcal{S}_2)= |(\mathcal{U}(\mathcal{S}_2) \setminus  \mathcal{V}_2(\mathcal{S}_2)) \setminus \mathcal{V}_1(\mathcal{S}_1)|
= |(\{3\}\setminus \emptyset)\setminus \{3\}| 
= 0.$
Therefore, using the above, we obtain the following upper bound on symmetric SDoF:
\begin{align*}
     &\textup{\textsf{SDoF}}^{\textup{\textsf{sym}}} \\&\leq  \min_{\mathcal{S}1, \mathcal{S}_2}\frac{|\mathcal{S}_1|+\mu (\mathcal{S}_1, \mathcal{S}_2)}{|\mathcal{S}_1|+\sum_{p=1}^{|\mathcal{S}_1|}\mathbf{1} (\mathcal{I}_{r_p}\neq \emptyset)+\sum_{q=1}^{|\mathcal{S}_2|}\mathbf{1} (\mathcal{I}_{m_q}\neq \emptyset)}
    \\&\leq \frac{1}{2}.
\end{align*}
Thus, for this network, the upper bound matches the lower bound. This implies that its exact symmetric SDoF is $\frac{1}{2}$.

Consider the $3$-user network of Fig. \ref{fig:Fig6} (h). At receiver $R\textsf{x}_1$, we have $\mathcal{I}_1=\emptyset$ and $\mathcal{T}_1=\{1\}\cup \mathcal{I}_1=\{1\}$.  At $R\textsf{x}_2$, we have $\mathcal{I}_2=\{1\}$ and $\mathcal{T}_2=\{2\}\cup \mathcal{I}_2=\{1, 2\}$.  At $R\textsf{x}_3$, we have $\mathcal{I}_3=\{1, 3\}$ and $\mathcal{T}_3=\{3\}\cup \mathcal{I}_3=\{1, 2, 3\}$.  Then, we can apply the scheme of Theorem \ref{TheoremGeneralTIM-CM-SecurePartition} by creating a secure partition $\mathcal{P}=\{\mathcal{P}_1=\{1\}, \mathcal{P}_2=\{2\},  \mathcal{P}_3=\{3\}\}$. In turn, this implies that we can schedule the transmitter in  $\mathcal{P}_1$ to  send its message over the first time slot (while $T\textsf{x}_2$ acts as a cooperative jammer by sending an artificial noise symbol and $T\textsf{x}_3$ remains silent over the first time slot, i.e., $\mathcal{C}_1=\{2\}$) and schedule the transmitter in  $\mathcal{P}_2$ to send its message over the second time slot (while $T\textsf{x}_1$ remains silent and $T\textsf{x}_3$  acts as a cooperative jammer by sending an artificial noise symbol over the second time slot, i.e., $\mathcal{C}_2=\{3\}$). Over the third time slot, we schedule the transmitter in $\mathcal{P}_3$ to send its message (while $T\textsf{x}_1$ and $T\textsf{x}_2$ remain silent, i.e., $\mathcal{C}_3=\emptyset$). Therefore, we achieve a symmetric SDoF of $\frac{1}{|\mathcal{P}|}=\frac{1}{3}$ per user.

Finding the upper bound on symmetric SDoF for  the network of Fig. \ref{fig:Fig6} (h) by applying the result of Theorem \ref{Theorem4UpperBound1} leads to a symmetric SDoF of $\frac{1}{2}$, which is greater than the achieved above lower bound. Since this network contains fractional signal generators (as in Definition \ref{FractionalGenerators}),  we can improve the upper bound by applying the result of Theorem \ref{Theorem5UpperBound2}. As we show next, this leads to an upper bound that matches the above lower bound.

Let us, for example, pick two disjoint subsets of receivers $\mathcal{S}_1=\{3\}$ and $\mathcal{S}_2=\{1\}$.  According to Definition \ref{FractionalGenerators} of fractional signal generators, Receiver $3\in\mathcal{S}_1$ (which has the interference set $\mathcal{I}_2=\{1, 2\}$), can generate the set $\mathcal{G}^{(1)}_3=\{1, 2\}$, with a permutation sequence $\Pi^{(3)}=(2, 1)$. Therefore, this leads to the multi-set $\mathcal{G}^{\mathcal{S}_1}=\mathcal{G}^{(1)}_3=\{1, 2\}$, which is the collection of all signal sets generated by the receivers in $\mathcal{S}_1$.  Following a similar analogy, Receiver $1\in\mathcal{S}_2$ (which has the interference set $\mathcal{I}_1=\emptyset$) does not generate any signals sets, i.e., $\mathcal{G}^{(2)}_1=\emptyset$. Therefore, $\mathcal{G}^{\mathcal{S}_2}=\mathcal{G}^{(2)}_1=\emptyset$. This implies that $\sum_{p=1}^{|\mathcal{S}_1|}|\mathcal{G}^{(1)}_{r_p}|=\sum_{r_p\in\mathcal{S}_1=\{3\}}|\mathcal{G}^{(1)}_{r_p}|=2$, whereas  $\sum_{q=1}^{|\mathcal{S}_2|}|\mathcal{G}^{(2)}_{m_q}|=\sum_{m_q\in\mathcal{S}_2=\{1\}}|\mathcal{G}^{(2)}_{m_q}|=0$. Moreover, we can, for example, pick the multi-set $\mathcal{U}(\mathcal{S}_2)$  and the two sets $\mathcal{\tilde{V}}(\mathcal{S}_2)$ and $\mathcal{\tilde{V}}_1(\mathcal{S}_1))$ as follows: $\mathcal{U}(\mathcal{S}_2)=\{\mathcal{T}_1\}=\{1\}$. $\mathcal{\tilde{V}}(\mathcal{S}_2)=\{\Pi^{(m_{|\mathcal{S}_2|})}_{|\mathcal{G}^{(2)}_{m_{|\mathcal{S}_2|}}|}\}=\emptyset$ because $\mathcal{G}^{(2)}_1=\emptyset$. $\mathcal{\tilde{V}}_1(\mathcal{S}_1))=\{\Pi^{(r_{|\mathcal{S}_1|})}_{|\mathcal{G}^{(1)}_{r_{|\mathcal{S}_1|}}|}\}=\{\Pi^{(3)}_{|\mathcal{G}^{(1)}_3|}\}=\{1\}$, where $1\in\mathcal{G}^{(1)}_3=\mathcal{I}_3$. Thus, using its definition, we can calculate $\tilde{\mu} (\mathcal{S}_1, \mathcal{S}_2)$ as follows: $\tilde{\mu} (\mathcal{S}_1, \mathcal{S}_2)=|(\mathcal{U}(\mathcal{S}_2) \setminus  \mathcal{\tilde{V}}_2(\mathcal{S}_2)) \setminus \mathcal{\tilde{V}}_1(\mathcal{S}_1))|
=|(\{1\}\setminus \emptyset)\setminus \{1\}|
= 0.$ Therefore, using the above, we obtain the following upper bound on symmetric SDoF:
\begin{align*}
     &\textup{\textsf{SDoF}}^{\textup{\textsf{sym}}} \\&\leq  \min_{(\mathcal{S}1, \mathcal{S}_2)} \min_{(\mathcal{G}^{\mathcal{S}_1}, \mathcal{G}^{\mathcal{S}_2})}\frac{|\mathcal{S}_1|+\tilde{\mu} (\mathcal{S}_1, \mathcal{S}_2)}{|\mathcal{S}_1|+\sum_{p=1}^{|\mathcal{S}_1|}|\mathcal{G}^{(1)}_{r_p}|+\sum_{q=1}^{|\mathcal{S}_2|}|\mathcal{G}^{(2)}_{m_q}|}
     \\&\leq \frac{1}{3}.\nonumber
\end{align*}
Thus, for this network, the upper bound matches the lower bound. This implies that its exact symmetric SDoF is $\frac{1}{3}$. 


\section{Conclusion}\label{Conclusion}

In this paper, we studied the impact of network topology on symmetric SDoF for the $K$-user partially connected interference networks with confidential messages in the absence of CSIT, except for the knowledge of network topology and channel statistics, a setting we named the general TIM-CM problem. We first presented necessary and sufficient conditions for the feasibility of positive symmetric SDoF for any topology. We then presented inner bounds on symmetric SDoF for the TIM-CM problem. For the first lower bound, we introduced and demonstrated how to utilize the concept of secure partition. For the second one, we introduced and showed how to utilize the concept of secure independent sets. We also presented outer bounds on symmetric SDoF. First, we showed how to collectively manage the interference components of the received signals at all receivers in order to obtain the first nontrivial upper bound. Then, through our second outer bound, we demonstrated that the first one can be made tighter in the presence of fractional signal generators. We also provided several examples to  further clarify the principles behind each of our proposed results and their application in calculating symmetric SDoF. In the end, we presented case studies through which we characterized the optimal symmetric SDoF for all $K$-user TIM-CM topologies with $K\leq 3$.  To the best of our knowledge, this is the first work to do a comprehensive study of the general TIM-CM problem. 

This work opens up several avenues for future research directions. The main one is that of finding the exact symmetric SDoF for the general TIM-CM problem. Another avenue is to investigate the topology conditions under which the secure partition based transmission scheme is better than the secure independent sets scheme in terms of spectral efficiency or vice versa. 

\section*{Appendix}
\begin{appendices}

\section{Proof of Theorem \ref{TheoremSDoFSymmetry}}\label{appendic:Theorem1Proof}
In this Section, we present the proof of Theorem \ref{TheoremSDoFSymmetry}: Its converse is provided in Appendix \ref{appendix:Thoerem1Converse}, whereas the achievability is provided in Appendix \ref{appendix:Thoerem1Achievability}.
\subsection{Converse Proof}
\label{appendix:Thoerem1Converse}

Consider the $K$-user TIM-CM network and suppose there exists a pair of users $i, j\in\{1, 2, \dots, K\}$, for $i\neq j$, such that $(i)$ $i\in\mathcal{I}_{j}$, $(ii)$ $j \in\mathcal{I}_{i}$, and $(iii)$ $\mathcal{I}_{i}\setminus \{j\} \subseteq \mathcal{I}_{j}\setminus{\{i\}}$, i.e., satisfying the conditions of Theorem \ref{TheoremSDoFSymmetry}.  The  signals seen at receivers $R\textsf{x}_{i}$ and $R\textsf{x}_{j}$ at time $t$ can then be written as follows:
\begin{align}
\label{ReceivedSignalatTxi-Rxi}
&Y_i(t)=h_{ii}X_i(t)+\sum_{m\in\mathcal{I}_i}h_{im}X_m(t)+Z_i(t),\\
\label{ReceivedSignalatTxj-Rxj}
&Y_j(t)=h_{jj}X_j(t)+\sum_{m\in\mathcal{I}_j}h_{jm}X_m(t)+Z_j(t).
\end{align}
We next show that since the considered user pair $i, j$ satisfies the  conditions of Theorem \ref{TheoremSDoFSymmetry}, then positive symmetric SDoF is not feasible, i.e, it is zero. To this end, we now create a virtual receiver $\tilde{R}\textsf{x}_{j}$ which is a genie-enhanced version of the original receiver $R\textsf{x}_{j}$. Note that $\mathcal{I}_j=(\mathcal{I}_j\setminus\{i\})\cup\{i\}$, and by condition $(iii)$  $\mathcal{I}_i\setminus\{j\} \subseteq \mathcal{I}_j\setminus\{i\}$. Therefore, $(\mathcal{I}_i\setminus\{j\})\cup\{i\}\subseteq( \mathcal{I}_j\setminus\{i\})\cup \{i\}=\mathcal{I}_j$. Since $(\mathcal{I}_i\setminus\{j\})\cup\{i\}$ is a subset of $\mathcal{I}_j$, we can let receiver $\tilde{R}\textsf{x}_{j}$ be as shown next.  Consider the received signal at $R\textsf{x}_{j}$ and let a genie provide the interfering signals set $X^n_{\mathcal{I}_j\setminus \mathcal{I}_i}$ (i.e., corresponding to the interference index set $\mathcal{I}_j\setminus\mathcal{I}_i$) to receiver $R\textsf{x}_{j}$. This leads to an  enhanced receiver $\tilde{R}\textsf{x}_{j}$ with the following output at time $t$:
\begin{align}
\label{EnhcancedRxj}
&\tilde{Y}_j(t)=h_{jj}X_j(t)+\sum_{m\in(\mathcal{I}_i\setminus\{j\})\cup \{i\}}h_{jm}X_m(t)+Z_j(t).
\end{align}
We note here that, if receiver $R\textsf{x}_j$ can decode so can receiver $\tilde{R}\textsf{x}_{j}$ because $(\mathcal{I}_i\setminus\{j\})\cup\{i\}$ is a subset of $\mathcal{I}_j$. The above enhancement leads to a statistically equivalent receive pair $(R\textsf{x}_i, \tilde{R}\textsf{x}_{j})$ that can be represented by a $2$-user symmetric interference network defined by:
\begin{align}
&Y_{i}(t)=\sum_{m\in\mathcal{T}_i}h_{i  m}X_{m}(t)+Z_{i}(t)\\
&\tilde{Y}_{j}(t)=\sum_{m\in\tilde{\mathcal{T}}_j}h_{j  m}X_{m}(t)+Z_{j}(t),
\end{align}
where $\tilde{\mathcal{T}}_j=\{j\}\cup(\mathcal{I}_i\setminus\{j\})\cup \{i\}=\mathcal{T}_i$. To complete the proof, we use the remaining steps to show that the SDoF at the original receiver $R\textsf{x}_{j}$ is upper bounded by the SDoF at the enhanced receiver $\tilde{R}\textsf{x}_{j}$, which is zero. To this end, we start by upperbounding the secrecy rate $R_j$ at $R\textsf{x}_{j}$  as follows:
\begin{align}
\label{RxjRatej11}
&n R_{j}=H(W_{j}|\mathcal{H})\\
\label{RxjRatej21}
&=I(W_{j}; Y^n_{j}|\mathcal{H})+H(W_{j}|Y^n_{j}, \mathcal{H})\\
\label{RxjRatej31}
&=I(W_{j}; Y^n_{j}|\mathcal{H})+no(n),\\
\label{RxjRatej41}
&\leq I(W_{j}; Y^n_{j}, X^n_{\mathcal{I}_j\setminus \mathcal{I}_i}|\mathcal{H})+no(n)\\
\label{RxjRatej51}
&= I(W_{j}; X^n_{\mathcal{I}_j\setminus \mathcal{I}_i}|\mathcal{H})+ I(W_{j}; Y^n_{j}|X^n_{\mathcal{I}_j\setminus \mathcal{I}_i}, \mathcal{H})+no(n)\\
\label{RxjRatej61}
&=I(W_{j}; Y^n_{j}|X^n_{\mathcal{I}_j\setminus \mathcal{I}_i}, \mathcal{H})+no(n)\\
\label{RxjRatej71}
&= I(W_{j}; \tilde{Y}^n_{j}|\mathcal{H})+no(n)\\
\label{RxjRatej81}
&= I(W_{j}; Y^n_{i}|\mathcal{H})+no(n)\\
\label{RxjRatej91}
&=no(n),
\end{align}
where  \eqref{RxjRatej31} is due to the decodability constraint in \eqref{DecodabilityConstraint}, \eqref{RxjRatej41} is due to the enhancement by a genie providing the interfering signals set $X^n_{\mathcal{I}_j\setminus \mathcal{I}_i}$ to receiver $R\textsf{x}_{j}$, and \eqref{RxjRatej51} follows from the chain rule of mutual information.  Equation \eqref{RxjRatej61} is due to the independence of $W_j$ from $X^n_{\mathcal{I}_j\setminus \mathcal{I}_i}$, whereas \eqref{RxjRatej71} follows from the fact that $\tilde{Y}^n_j$ is the enhanced version of $Y^n_j$. Equation \eqref{RxjRatej81} is due to the fact that the joint distribution of $(W_j, \tilde{Y}^n_j)$ is the same as the joint distribution of $(W_j, Y^n_i)$, whereas \eqref{RxjRatej91} is due to the secrecy constraint in \eqref{ConfidentialityConstraint}. Taking the limit of \eqref{RxjRatej91} as $n\rightarrow \infty$, we get $R_j=0$. This implies that the symmetric secrecy rate is zero. Therefore, by Definition \ref{SDoF_symDefinition},   symmetric SDoF is also zero. \QEDA

\subsection{Achievability Proof}
\label{appendix:Thoerem1Achievability}

Our aim is to show that symmetric SDoF is nonzero, if the conditions of Theorem \ref{TheoremSDoFSymmetry} are not met. Consider the $K$-user TIM-CM network and suppose there is no single pair of users $i, j\in\{1, 2, \dots, K\}$, for $i\neq j$, within the network that satisfies the conditions of Theorem \ref{TheoremSDoFSymmetry}. We next provide a simple secure TDMA algorithm that achieves SDoF of $\frac{1}{K}$.


The proposed secure transmission works as follows:  Let $\mathcal{C}_t$ be the index set of transmitters that have to act as cooperative jammers (by sending artificial noise symbols) during the $t$th time slot. The following steps show how each transmitter $T\textsf{x}_k$, for $k\in\{1, 2, \dots, K\}$,  can securely transmit its message $W_k$ during time slot $t\in\{1, 2, \dots, T\}$: 
\begin{enumerate}[(i)]
\item Let  $T\textsf{x}_k$ send $W_k$ over the $t$th time slot.
\item Let  all transmitters in $\{\{T\textsf{x}_i\}_{i\neq k}^K$:  $i\in\mathcal{C}_t\}$ send artificial noise symbols during time slot $t$. 
\item Let all transmitters $\{\{T\textsf{x}_i\}_{i\neq k}^K$: $i\notin\mathcal{C}_t\}$ remain silent during time slot $t$. 
\item Repeat steps $(i)-(iii)$ for a total of $T=K$ times.
\end{enumerate}
We are thus able to separately serve the $K$ users  over a transmission block of length $T=K$ time slots, which leads to a symmetric SDoF of $\frac{1}{K}$ per user. Note that $\frac{1}{K}$ is the least symmetric SDoF that can be achieved for any TIM-CM topology that is symmetric SDoF feasible. \QEDA

\section{Proof of Theorem \ref{Theorem4UpperBound1}}
\label{appendix:Theorem4UpperBound1Proof}

Consider the signal $Y^n_k$ received at $R\textsf{x}_k$ over transmission blocklength $n$. For simplicity of notation, we denote $\mathbf{H}^n_{ki}=h_{ki}\mathbf{I}_n$ as the channel matrix between transmitter $T\textsf{x}_i$ and receiver $R\textsf{x}_k$ over a block of length $n$, where $h_{ki}$ is the time-invariant channel coefficient and 
$\mathbf{I}_n$ is the $n\times n$ identity matrix. We invoke decodability and confidentiality constraints to upper bound the secrecy rate $R_k$ at $R\textsf{x}_k$ as follows:
\begin{align}
\label{UpperBound1_1}
&nR_k=H(W_k|\mathcal{H})\\
\label{UpperBound1_2}
&=I(W_k; Y^n_k|\mathcal{H})+H(W_k|Y^n_k, \mathcal{H})\\
\label{UpperBound1_3}
&= I(W_k; Y^n_k|\mathcal{H})+n o(n)\\
\label{UpperBound1_4}
&= I(W_k; Y^n_k|\mathcal{H})-I(W_{\mathcal{I}_k}; Y^n_k|\mathcal{H})\nonumber\\&\hspace{15pt}+I(\mathcal{W}_{\mathcal{I}_k}; Y^n_k|\mathcal{H})+n o(n)\\
\label{UpperBound1_5}
&=h(Y^n_k|W_{\mathcal{I}_k}, \mathcal{H})-h(Y^n_k|W_k, \mathcal{H})\nonumber\\&\hspace{15pt}+I(\mathcal{W}_{\mathcal{I}_k}; Y^n_k|\mathcal{H})+n o(n)\\
\label{UpperBound1_6}
&= h(Y^n_k|W_{\mathcal{I}_k}, \mathcal{H})\nonumber\\&\hspace{15pt}-h(Y^n_k|W_k, \mathcal{H})+n o(n)\\
\label{UpperBound1_61}
&\leq  h(Y^n_k|W_{\mathcal{I}_k}, \mathcal{H})\nonumber\\&\hspace{15pt}-h(Y^n_k|W_k, X^n_k, \mathcal{H})+n o(n)\\
\label{UpperBound1_7}
&= h(Y^n_k|W_{\mathcal{I}_k}, \mathcal{H})\nonumber\\&\hspace{15pt}-h(Y^n_k|X^n_k, \mathcal{H})+n o(n)\\
\label{UpperBound1_8}
&= h(Y^n_k|W_{\mathcal{I}_k}, \mathcal{H})\nonumber\\&\hspace{15pt}-h\left(\sum_{i\in\mathcal{T}_k}\mathbf{H}^n_{ki}X^n_i+Z^n_k\Big| X^n_k, \mathcal{H}\right)+n o(n)\\
\label{UpperBound1_9}
&=h(Y^n_k|W_{\mathcal{I}_k}, \mathcal{H})\nonumber\\&\hspace{15pt}-h\left(\sum_{i\in\mathcal{I}_k}\mathbf{H}^n_{ki}X^n_i+Z^n_k\Big|\mathcal{H}\right)+no(n)\\
\label{UpperBound1_91}
&\leq   h(Y^n_k|W_{\mathcal{I}_k}, \mathcal{H})\nonumber\\&\hspace{15pt}-h\left(\sum_{i\in\mathcal{I}_k}\mathbf{H}^n_{ki}X^n_i+Z^n_k\Big|X^n_{\mathcal{I}_k\setminus\{i\}}, \mathcal{H}\right)+no(n)\\
\label{UpperBound1_10}
&=h(Y^n_k|W_{\mathcal{I}_k}, \mathcal{H})-h(\mathbf{H}^n_{ki}X^n_{i}+Z^n_k|\mathcal{H})+no(n),
\end{align}
for all $i\in\mathcal{I}_k$. Here, \eqref{UpperBound1_3} is due to Fano's inequality, \eqref{UpperBound1_5} comes from the identity $I(X; Y)-I(X; Z)=h(X|Z)-h(X|Y)$, \eqref{UpperBound1_6} is due to the secrecy constraint in \eqref{ConfidentialityConstraint}, \eqref{UpperBound1_61} is due to the fact that conditioning reduces entropy,  \eqref{UpperBound1_7} follows from the Markov chain $W_k\rightarrow X^n_k\rightarrow Y^n_k$, and \eqref{UpperBound1_91} is due to the fact that conditioning reduces entropy. Starting with \eqref{UpperBound1_10}, we can then proceed as follows. 
\begin{align}
\label{Type1Bound_1New1}
nR_k& \leq h(Y^n_k|W_{\mathcal{I}_k}, \mathcal{H})-h(\mathbf{H}^n_{ki}X^n_{i}+Z^n_k|\mathcal{H})+no(n)\nonumber\\
&= h(Y^n_k|W_{\mathcal{I}_k}, \mathcal{H})-h(\tilde{Y}^n_{i}| \mathcal{H})+no(n)\\
\label{Type1Bound_1New2}
&=h(Y^n_k|W_{\mathcal{I}_k}, \mathcal{H})-I(W_i; \tilde{Y}^n_i|\mathcal{H})\nonumber\\&\hspace{15pt}-h(\tilde{Y}^n_i|W_i, \mathcal{H})+no(n)\\
\label{Type1Bound_1New3}
&=h(Y^n_k|W_{\mathcal{I}_k}, \mathcal{H})-H(W_i|\mathcal{H})+H(W_i|\tilde{Y}^n_i, \mathcal{H})\nonumber\\&\hspace{15pt}-h(\tilde{Y}^n_i|W_i, \mathcal{H})+no(n)\\
\label{Type1Bound_1New4}
&=h(Y^n_k|W_{\mathcal{I}_k}, \mathcal{H})-nR_i+H(W_i|\tilde{Y}^n_i, \mathcal{H})\nonumber\\&\hspace{15pt}-h(\tilde{Y}^n_i|W_i, \mathcal{H})+no(n)\\
\label{Type1Bound_1New5}
&=h(Y^n_k|W_{\mathcal{I}_k}, \mathcal{H})-nR_i-h(\tilde{Y}^n_i|W_i, \mathcal{H})+no(n),
\end{align}
for all $i\in\mathcal{I}_k$ and where \eqref{Type1Bound_1New1} follows from substitution  $\tilde{Y}^n_i=\mathbf{H}^n_{ki}X^n_{i}+Z^n_k$. Here,  \eqref{Type1Bound_1New5} is due to the  fact that $\tilde{Y}^n_i$ is the enhanced version of the signal $Y^n_i$ seen at receiver $R\textsf{x}_i$, and therefore due to the decodability constraint, $H(W_i|\tilde{Y}^n_i, \mathcal{H})=o(n)$. 

We now bound \eqref{Type1Bound_1New5} in two different ways. We start with the first type of bound as follows: 
\begin{align}
\label{Type1Bound_1New6}
n(R_i+R_k) &\leq h(Y^n_k|W_{\mathcal{I}_k}, \mathcal{H})-h(\tilde{Y}^n_i|W_i, \mathcal{H})+no(n)\\
\label{Type1Bound_1New7}
 & \leq h(Y^n_k|\mathcal{H})-h(\tilde{Y}^n_i|W_i, \mathcal{H})+no(n)\\
\label{Type1Bound_1New8}
&\leq  n\log(P)-h(\tilde{Y}^n_i|W_i, \mathcal{H})\nonumber\\&\hspace{15pt}+no(n)+no(\log(P)),
\end{align}
for all $i\in\mathcal{I}_k$. Here, \eqref{Type1Bound_1New6} follows from rearranging the terms in \eqref{Type1Bound_1New5}, \eqref{Type1Bound_1New7} is due to the fact that conditioning reduces entropy, and \eqref{Type1Bound_1New8} follows from the fact that Gaussian distribution maximizes entropy \cite{CoverThomas}.   We note here that, for any receiver $R\textsf{x}_k$ with interference set $\mathcal{I}_k=\emptyset$, we would get the following bound instead of \eqref{Type1Bound_1New8}: $nR_k \leq n\log(P)+no(\log(P))$.

We next derive the second type of bound. We start by bounding \eqref{Type1Bound_1New6} as follows:
\begin{align}
n(R_i+R_k) &\leq h(Y^n_k|W_{\mathcal{I}_k}, \mathcal{H})-h(\tilde{Y}^n_i|W_i, \mathcal{H})+no(n)\nonumber\\
\label{Type1Bound_1New10}
\leq nR_k&+\sum_{j\in\mathcal{T}_k}h(\tilde{Y}^n_j|W_j, \mathcal{H})-h(\tilde{Y}^n_{i}|W_{i}, \mathcal{H})\nonumber\\&\hspace{15pt}+no(\log(P))+no(n),
\end{align}
for all $i\in\mathcal{I}_k$ and where $\mathcal{T}_k=\{k\}\cup\mathcal{I}_k$. Here, \eqref{Type1Bound_1New10} follows from bounding  $h(Y^n_k|W_{\mathcal{I}_k}, \mathcal{H})$ using Lemma \ref{Lemma1} below whose proof is in Appendix \ref{appendix:Lemma1Proof}.  \begin{lemma}\label{Lemma1} Let $\tilde{Y}^n_i=\mathbf{H}^n_{ki}X^n_i+Z^n_k$ , for all $i\in\mathcal{T}_k=\{k\}\cup\mathcal{I}_k$, then we have the following:
\begin{align}
\label{Lemma1EquationB}
h(Y^n_k|W_{\mathcal{I}_k}, \mathcal{H})&\leq nR_k+\sum\limits_{j\in\mathcal{T}_k}h(\tilde{Y}^n_j|W_j, \mathcal{H})\nonumber\\&\hspace{15pt}+no(\log(P)).
\end{align}
\end{lemma}
By simplifying \eqref{Type1Bound_1New10}, we then obtain the desired second type of bound:
\begin{align}
\label{Type1Bound_1New11}
nR_i&\leq \sum_{j\in\mathcal{T}_k}h(\tilde{Y}^n_j|W_j, \mathcal{H})-h(\tilde{Y}^n_{i}|W_{i}, \mathcal{H})\nonumber\\&\hspace{15pt}+no(\log(P))+no(n).
\end{align}
Similarly, we note that for any $R\textsf{x}_k$ with $\mathcal{I}_k=\emptyset$, we would get the following bound instead of \eqref{Type1Bound_1New11}: $0 \leq\sum_{j\in \mathcal{T}_k} h(\tilde{Y}^n_{j}|W_j,  \mathcal{H})=h(\tilde{Y}^n_{k}|W_k,  \mathcal{H})$.

For the purpose of the current paper, we are interested in symmetric SDoF. Therefore, we can now drop the indexing and instead use the notation $R_i=R^{\textsf{sym}}, \forall i\in\{1, 2, \dots, K\}$, which represents the symmetric secrecy rate at any receiver. Moreover, as  shown by \eqref{Type1Bound_1New8} and  \eqref{Type1Bound_1New11}, we get two types of bounds on $R^{\textsf{sym}}$. Consider the receiver $R\textsf{x}_k$. Then, from \eqref{Type1Bound_1New8}, we directly obtain the following  first type of bound:
\begin{align}
\label{Type1Bound_2} 
    2nR^{\textsf{sym}}&\leq n\log(P) - h(\tilde{Y}^n_{i}|W_i, \mathcal{H})\nonumber\\&\hspace{15pt}+no(n)+no(\log(P)),
\end{align}
for all $i\in \mathcal{I}_k$. We again note here that, for any receiver $R\textsf{x}_k$ with $\mathcal{I}_k=\emptyset$, we would get $nR^{\textsf{sym}}\leq n\log(P)+no(\log(P))$ instead of \eqref{Type1Bound_2}. The second type of bound follows from \eqref{Type1Bound_1New11} and is given by:
\begin{align}
\label{Type2Bound_2} 
    nR^{\textsf{sym}}&\leq \sum_{j\in \mathcal{T}_k} h(\tilde{Y}^n_{j}|W_j, \mathcal{H}) - h(\tilde{Y}^n_{i}|W_i, \mathcal{H})\nonumber\\&\hspace{15pt}+n o(\log(P))+no(n),
\end{align}
for all $i\in \mathcal{I}_k$ and where $\mathcal{T}_k=\{k\}\cup\mathcal{I}_k$. Similarly, we note that, for any $R\textsf{x}_k$ with $\mathcal{I}_k=\emptyset$, we would get $ 0 \leq\sum_{j\in \mathcal{T}_k} h(\tilde{Y}^n_{j}|W_j,  \mathcal{H})=h(\tilde{Y}^n_{k}|W_k,  \mathcal{H})$ instead of \eqref{Type2Bound_2}. 

We now remark the following fact which is heavily used in the subsequent steps of this proof.
\begin{remark} \label{remark33}
 Due to the absence of CSI at the transmitters, i.e., no CSIT under the TIM-CM framework, and the fact that the channel coefficients are i.i.d., the enhanced signal $\tilde{Y}^n_i=\mathbf{H}^n_{ki}X^n_i+Z^n_k$ is statistically equivalent to the enhanced signal $\tilde{Y}^{'n}_i=\mathbf{H}^n_{ji}X^n_i+Z^n_j$, for all $k\neq j\in\{1, 2, \dots, K\}$.
\end{remark}

We can now proceed as follows: Pick any two arbitrary and disjoint subsets of receivers $\mathcal{S}_1$ and $\mathcal{S}_2$, i.e., such that $\mathcal{S}_1\cap \mathcal{S}_2=\emptyset$  and $\mathcal{S}_1, \mathcal{S}_2 \subseteq \{1, 2, \dots, K\}$. For all receivers in $\mathcal{S}_1$, we can apply  the first type of bound \eqref{Type1Bound_2}, and for all receivers in $\mathcal{S}_2$, we apply the second type of bound \eqref{Type2Bound_2}. Then, we can add up the resulting $|\mathcal{S}_1|+|\mathcal{S}_2|$ equations and see the resulting cancellation possible as a function of the set choices $(\mathcal{S}_1, \mathcal{S}_2)$. Moreover, suppose $\mathcal{S}_1=\{r_1, r_2, \dots, r_{|\mathcal{S}_1|}\}$ and $\mathcal{S}_2=\{m_1, m_2, \dots, m_{|\mathcal{S}_2|}\}$. Thus, by adding up all the equations and rearranging terms, we obtain the following:
\begin{align}
\label{RateSum}
&nR^{\textsf{sym}}\left(|\mathcal{S}_1|+\sum_{p=1}^{|\mathcal{S}_1|}\mathbf{1} (\mathcal{I}_{r_p}\neq \emptyset)\right)\nonumber\\&\hspace{25pt}+nR^{\textsf{sym}}\sum_{q=1}^{|\mathcal{S}_2|}\mathbf{1} (\mathcal{I}_{m_q}\neq \emptyset)\nonumber\\&\hspace{15pt}\leq|\mathcal{S}_1| n\log(P)+no(\log(P))+no(n)\nonumber\\&\hspace{25pt}+\underbrace{\sum\limits_{m_q\in\mathcal{S}_2}\sum\limits_{i\in\mathcal{T}_{m_q}}h(\tilde{Y}^n_{i}|W_i, \mathcal{H})}_{\mathcal{U}(\mathcal{S}_2)}\nonumber\\&\hspace{25pt}-\hspace{-8pt}\underbrace{\sum\limits_{j_q\in\mathcal{I}_{m_q}:m_q\in\mathcal{S}_2}\hspace{-8pt}h(\tilde{Y}^n_{j_q}|W_{j_q}, \mathcal{H})}_{\mathcal{V}_2(\mathcal{S}_2)} \nonumber\\&\hspace{25pt}-\hspace{-8pt}\underbrace{\sum\limits_{\ell_p\in\mathcal{I}_{r_p}:r_p\in\mathcal{S}_1}\hspace{-8pt}h(\tilde{Y}^n_{\ell_p}|W_{\ell_p}, \mathcal{H})}_{\mathcal{V}_1(\mathcal{S}_1)}.
\end{align}

Starting with the last three terms in \eqref{RateSum}, it is clear that the aim should be to choose the indexes for signal components corresponding to the set choices $\mathcal{S}_1$ and  $\mathcal{S}_2$ in a way that allows the maximum cancellation of the entropy terms, while taking advantage of the signal distribution properties stated in Remark \ref{remark33}. Depending upon the network topology $\mathcal{G}$, this cancellation will then result in zero or more positive entropy terms. We calculate the number of these remaining positive entropy terms as shown next.  

To calculate the number of (positive) entropy terms resulting from $\sum_{m_q\in\mathcal{S}_2}\sum_{i\in\mathcal{T}_{m_q}}h(\tilde{Y}^n_{i}|W_i,  \mathcal{H})$, we can use the cardinality of the collection of all index sets of transmitters whose signals are respectively seen at the receivers in $\mathcal{S}_2=\{m_1, m_2, \dots, m_{|\mathcal{S}_2|}\}$. We denote this multi-set of transmitter indexes as $\mathcal{U}(\mathcal{S}_2)=\{\mathcal{T}_{m_q}\}_{q=1}^{|\mathcal{S}_2|}$. To calculate the number of (negative) entropy terms resulting from $\sum_{j_q\in\mathcal{I}_{m_q}:m_q\in\mathcal{S}_2}h(\tilde{Y}^n_{j_q}|W_{j_q},  \mathcal{H})$, we use the cardinality of the set denoted by $\mathcal{V}_2(\mathcal{S}_2)=  \{j_{1}, j_{2}, \dots, j_{|\mathcal{S}_2|}\}$, which is an arbitrary set of indexes $j_{q}\in \mathcal{I}_{m_q}$ for $m_q\in \mathcal{S}_2$.  Similarly, to calculate the number of (negative) entropy terms resulting from $\sum_{\ell_p\in\mathcal{I}_{r_p}:r_p\in\mathcal{S}_1}h(\tilde{Y}^n_{\ell_p}|W_{\ell_p},  \mathcal{H})$, we use the cardinality of the set denoted by $\mathcal{V}_1(\mathcal{S}_1)= \{\ell_{1}, \ell_{2}, \dots, \ell_{|\mathcal{S}_1|}\}$, which is an arbitrary set of indexes $\ell_{p}\in \mathcal{I}_{r_p}$ for $r_p\in \mathcal{S}_1=\{r_1, r_2, \dots, r_{|\mathcal{S}_1|}\}$.

 From the defined above three sets, we can then calculate the resulting number of (positive) entropy terms as follows:
\begin{align}
\label{EntropySetCounter}
\mu (\mathcal{S}_1, \mathcal{S}_2) = \min_{\mathcal{V}_2(\mathcal{S}_2), \mathcal{V}_1(\mathcal{S}_1)}|(\mathcal{U}(\mathcal{S}_2) \setminus  \mathcal{V}_2(\mathcal{S}_2)) \setminus \mathcal{V}_1(\mathcal{S}_1)|.
\end{align}

As indicated in \eqref{PowerConstraint},  each signal is transmitted with an average power constraint $P$. Thus, each of the $\mu (\mathcal{S}_1, \mathcal{S}_2)$ entropy terms can be upper bounded by $n\log(P)$. Therefore, using \eqref{RateSum} and \eqref{EntropySetCounter}, we can upper bound the secrecy rate as follows:
\begin{align}
\label{BlockSymmetricRate}
     &nR^{\textsf{sym}}\nonumber\\&\leq \frac{|\mathcal{S}_1|n\log(P)+\mu (\mathcal{S}_1, \mathcal{S}_2)n\log(P)}{|\mathcal{S}_1|+\sum_{p=1}^{|\mathcal{S}_1|}\mathbf{1} (\mathcal{I}_{r_p}\neq \emptyset)+\sum_{q=1}^{|\mathcal{S}_2|}\mathbf{1} (\mathcal{I}_{m_q}\neq \emptyset)}\nonumber\\&+\frac{no(\log(P))+no(n)}{|\mathcal{S}_1|+\sum_{p=1}^{|\mathcal{S}_1|}\mathbf{1} (\mathcal{I}_{r_p}\neq \emptyset)+\sum_{q=1}^{|\mathcal{S}_2|}\mathbf{1} (\mathcal{I}_{m_q}\neq \emptyset)}.
\end{align}
Then, by dividing both sides of  \eqref{BlockSymmetricRate} by $n\log(P)$, minimizing over all choices of $(\mathcal{S}_1, \mathcal{S}_2)$, and taking the limit as $n\rightarrow \infty$ and $P\rightarrow \infty$, we obtain the following upper bound on symmetric SDoF:
\begin{align}
\label{SDoFProofExpressionThm4}
&\textup{\textsf{SDoF}}^{\textup{\textsf{sym}}}\nonumber\\&\leq \min_{\mathcal{S}1, \mathcal{S}_2}\frac{|\mathcal{S}_1|+\mu (\mathcal{S}_1, \mathcal{S}_2)}{|\mathcal{S}_1|+\sum_{p=1}^{|\mathcal{S}_1|}\mathbf{1} (\mathcal{I}_{r_p}\neq \emptyset)+\sum_{q=1}^{|\mathcal{S}_2|}\mathbf{1} (\mathcal{I}_{m_q}\neq \emptyset)},
\end{align}
which completes the proof of Theorem \ref{Theorem4UpperBound1}. \QEDA
\section{Proof of Theorem \ref{Theorem5UpperBound2}}
\label{appendix:Theorem5UpperBound2Proof}

We now present the proof of Theorem \ref{Theorem5UpperBound2}. As it will be demonstrated by the steps below, the main difference between Theorem \ref{Theorem4UpperBound1} and Theorem \ref{Theorem5UpperBound2} lies in the way, for a given receiver $R\textsf{x}_k$, we lower bound the entropy of the interfering signals' component $h\left(\sum_{i\in\mathcal{I}_k}\mathbf{H}^n_{ki}X^n_i+Z^n_k| \mathcal{H}\right)$ from \eqref{UpperBound1_9} when receiver $R\textsf{x}_k$ is  fractional signal generator as given by Definition \ref{FractionalGenerators}.

Consider the signal $Y^n_k$ received at $R\textsf{x}_k$ over transmission blocklength $n$. Starting with \eqref{UpperBound1_9}, we can upper bound the symmetric secrecy rate $R^{\textsf{sym}}$ as follows:
\begin{align}
&nR^{\textsf{sym}}\leq h(Y^n_k|W_{\mathcal{I}_k},  \mathcal{H})\nonumber\\&\hspace{35pt}-h\left(\sum_{i\in\mathcal{I}_k}\mathbf{H}^n_{ki}X^n_i+Z^n_k\Big| \mathcal{H}\right)+no(n)\nonumber\\
\label{UpperBound2_2}
\hspace{15pt}&\leq h(Y^n_k|W_{\mathcal{I}_k}, \mathcal{H})-n |\mathcal{G}_k|R^{\textsf{\textup{sym}}}\nonumber\\&\hspace{35pt}-h(\tilde{Y}^n_{\Pi_{|\mathcal{G}_k|}}|W_{\Pi_{|\mathcal{G}_k|}},  \mathcal{H})+no(n),
\end{align}
where $\Pi_{|\mathcal{G}_k|}\in\mathcal{G}_k\subseteq \mathcal{I}_k$ and \eqref{UpperBound2_2} follows from Lemma \ref{Lemma2} below whose proof is in Appendix \ref{appendix:Lemma2Proof}.

\begin{lemma}\label{Lemma2}  Let $\tilde{Y}^n_i=\mathbf{H}^n_{ki}X^n_i+Z^n_k$ , for all $i\in\mathcal{T}_k=\{k\}\cup\mathcal{I}_k$, and consider any set $\mathcal{G}_k\subseteq\mathcal{I}_k$ and a permutation sequence  $\Pi^{(k)}=(\Pi_1, \Pi_2, \dots \Pi_{|\mathcal{G}_k|})$ of the elements of $\mathcal{G}_k$ satisfying Definition \ref{FractionalGenerators}, then
\begin{align}
\label{Lemma2EquationB}
&h\left(\sum_{i\in\mathcal{I}_k}\mathbf{H}^n_{ki}X^n_i+Z^n_k\Big|\mathcal{H}\right)\geq n |\mathcal{G}_k|R^{\textsf{\textup{sym}}}\nonumber\\&\hspace{35pt}+h(\tilde{Y}^n_{\Pi_{|\mathcal{G}_k|}}|W_{\Pi_{|\mathcal{G}_k|}},  \mathcal{H})+no(n).
\end{align}
\end{lemma}

We now bound \eqref{UpperBound2_2} in two different ways. We start with the first type of bound as follows:
\begin{align}
nR^{\textsf{sym}}&\leq h(Y^n_k|W_{\mathcal{I}_k}, \mathcal{H})-n |\mathcal{G}_k|R^{\textsf{\textup{sym}}}\nonumber\\&\hspace{35pt}-h(\tilde{Y}^n_{\Pi_{|\mathcal{G}_k|}}|W_{\Pi_{|\mathcal{G}_k|}},  \mathcal{H})+no(n) \nonumber\\
\label{UpperBound2_11}
&\leq h(Y^n_k|\mathcal{H})-n |\mathcal{G}_k|R^{\textsf{\textup{sym}}}\nonumber\\&\hspace{35pt}-h(\tilde{Y}^n_{\Pi_{|\mathcal{G}_k|}}|W_{\Pi_{|\mathcal{G}_k|}},  \mathcal{H})+no(n)\\
\label{UpperBound2_12}
&\leq n\log (P)-n |\mathcal{G}_k|R^{\textsf{\textup{sym}}}\nonumber\\&\hspace{35pt}-h(\tilde{Y}^n_{\Pi_{|\mathcal{G}_k|}}|W_{\Pi_{|\mathcal{G}_k|}},  \mathcal{H})\nonumber\\&\hspace{35pt}+no(n)+no(\log(P)),
\end{align}
where $\Pi_{|\mathcal{G}_k|}\in\mathcal{G}_k\subseteq \mathcal{I}_k$, \eqref{UpperBound2_11} is due to the fact that conditioning reduces entropy, and \eqref{UpperBound2_12} follows from the fact that Gaussian distribution maximizes entropy \cite{CoverThomas}. By rearranging the terms in \eqref{UpperBound2_12}, we then obtain the desired first type of bound:
\begin{align} 
\label{Type1Bound2_1}
&nR^{\textsf{sym}}(1+|\mathcal{G}_k|) \leq n\log(P)\nonumber\\&\hspace{35pt}-h(\tilde{Y}^n_{\Pi_{|\mathcal{G}_k|}}|W_{\Pi_{|\mathcal{G}_k|}},  \mathcal{H})+no(n),
\end{align}
where $\Pi_{|\mathcal{G}_k|}\in\mathcal{G}_k\subseteq \mathcal{I}_k$. We next derive the second type of bound.  To this end, we start with \eqref{UpperBound2_2} as follows:
\begin{align}
& nR^{\textsf{sym}}\leq h(Y^n_k|W_{\mathcal{I}_k}, \mathcal{H})-n |\mathcal{G}_k|R^{\textsf{\textup{sym}}}\nonumber\\&\hspace{35pt}-h(\tilde{Y}^n_{\Pi_{|\mathcal{G}_k|}}|W_{\Pi_{|\mathcal{G}_k|}},  \mathcal{H})+no(n) \nonumber\\
\label{UpperBound2_111}
&\hspace{15pt}\leq nR^{\textsf{sym}} + \sum_{i\in\mathcal{T}_k}h(\tilde{Y}^n_i|W_i,  \mathcal{H})+n o(\log(P))\nonumber\\&\hspace{0pt}-n |\mathcal{G}_k|R^{\textsf{\textup{sym}}}-h(\tilde{Y}^n_{\Pi_{|\mathcal{G}_k|}}|W_{\Pi_{|\mathcal{G}_k|}},  \mathcal{H})+no(n),
\end{align}
where $\Pi_{|\mathcal{G}_k|}\in\mathcal{G}_k\subseteq \mathcal{I}_k$ and \eqref{UpperBound2_111} follows from upperbounding $h(Y^n_k|W_{\mathcal{I}_k},  \mathcal{H})$ using Lemma \ref{Lemma1}.  By simplifying and rearranging the terms in \eqref{UpperBound2_111}, we then obtain the desired second type of bound:
\begin{align}
\label{Type2Bound2_1}
&nR^{\textsf{sym}} |\mathcal{G}_k| \leq \sum_{i\in\mathcal{T}_k}h(\tilde{Y}^n_i|W_i,  \mathcal{H})\nonumber\\&\hspace{5pt}-h(\tilde{Y}^n_{\Pi_{|\mathcal{G}_k|}}|W_{\Pi_{|\mathcal{G}_k}|},  \mathcal{H})+n o(\log(P))+no(n),
\end{align}
where $\Pi_{|\mathcal{G}_k|}\in\mathcal{G}_k\subseteq \mathcal{I}_k$.

We can now proceed as follows: Pick any two arbitrary and disjoint subsets of receivers $\mathcal{S}_1$ and $\mathcal{S}_2$, i.e., such that $\mathcal{S}_1\cap \mathcal{S}_2=\emptyset$  and $\mathcal{S}_1, \mathcal{S}_2 \subseteq \{1, 2, \dots, K\}$.  Suppose $\mathcal{S}_1=\{r_1, r_2, \dots r_{|\mathcal{S}_1|}\}$. Consider a collection of sets $\mathcal{G}^{\mathcal{S}_1}=\{\mathcal{G}^{(1)}_{r_1}, \mathcal{G}^{(1)}_{r_2}, \dots, \mathcal{G}^{(1)}_{r_{|\mathcal{S}_1|}}\}$ such that Receiver $r_p$ is a generator of $\mathcal{G}^{(1)}_{r_p}\subseteq \mathcal{I}_{r_p}$ for all $p\in\{1, 2, \dots, |\mathcal{S}_1|\}$, with a permutation $\Pi^{(r_p)}=\{\Pi^{(r_p)}_1, \Pi^{(r_p)}_2, \dots \Pi^{(r_p)}_{|\mathcal{G}^{(1)}_{r_p}|}\}$ according to Definition \ref{FractionalGenerators}. Similarly, suppose $\mathcal{S}_2=\{m_1, m_2, \dots m_{|\mathcal{S}_2|}\}$. Consider a collection of sets $\mathcal{G}^{\mathcal{S}_2}=\{\mathcal{G}^{(2)}_{m_1}, \mathcal{G}^{(2)}_{m_2}, \dots, \mathcal{G}^{(2)}_{m_{|\mathcal{S}_2|}}\}$ such that Receiver $m_q$ is a generator of $\mathcal{G}^{(2)}_{m_q}\subseteq \mathcal{I}_{m_q}$ for all $q\in\{1, 2, \dots, |\mathcal{S}_2|\}$, with a permutation $\Pi^{(m_q)}=\{\Pi^{(m_q)}_1, \Pi^{(m_q)}_2, \dots \Pi^{(m_q)}_{|\mathcal{G}^{(2)}_{m_q}|}\}$ according to Definition \ref{FractionalGenerators}.  Then, for all receivers in $\mathcal{S}_1$, we can apply  the first type of bound \eqref{Type1Bound2_1}, and for all receivers in $\mathcal{S}_2$, we apply the second type of bound \eqref{Type2Bound2_1}. Then, we add up the resulting $|\mathcal{S}_1|+|\mathcal{S}_2|$ equations, and see the resulting cancellation possible as a function of the choice of $(\mathcal{S}_1, \mathcal{S}_2, \mathcal{G}^{\mathcal{S}_1}, \mathcal{G}^{\mathcal{S}_2})$.  Thus, by adding up all the equations and rearranging terms, we obtain the following:
\begin{align}
\label{RateSum2}
n&R^{\textsf{sym}}\left(|\mathcal{S}_1|+\sum_{p=1}^{|\mathcal{S}_1|}|\mathcal{G}^{(1)}_{r_p}|\right)+nR^{\textsf{sym}}\sum_{q=1}^{|\mathcal{S}_2|}|\mathcal{G}^{(2)}_{m_q}|\nonumber\\&\leq |\mathcal{S}_1| n\log(P)+no(\log(P))+no(n)\nonumber\\&\hspace{15pt}+\underbrace{\sum\limits_{m_q\in\mathcal{S}_2}\sum\limits_{i\in\mathcal{T}_{m_q}}h(\tilde{Y}^n_{i}|W_i,  \mathcal{H})}_{\mathcal{U}(\mathcal{S}_2)}\nonumber\\&\hspace{15pt}-\hspace{-8pt}\underbrace{\sum\limits_{\Pi^{(m_q)}_{|\mathcal{G}^{(2)}_{m_q}|}\in\mathcal{I}_{m_q}:m_q\in\mathcal{S}_2}\hspace{-9pt}h\left(\tilde{Y}^n_{\Pi^{(m_q)}_{|\mathcal{G}^{(2)}_{m_q}|}}\Big|W_{\Pi^{(m_q)}_{|\mathcal{G}^{(2)}_{m_q}|}},  \mathcal{H}\right)}_{\mathcal{\tilde{V}}_2(\mathcal{S}_2)}\nonumber\\&\hspace{15pt}-\hspace{-8pt}\underbrace{\sum\limits_{\Pi^{(r_p)}_{|\mathcal{G}^{(1)}_{r_p}|}\in\mathcal{I}_{r_p}:r_p\in\mathcal{S}_1}\hspace{-9pt}h\left(\tilde{Y}^n_{\Pi^{(r_p)}_{|\mathcal{G}^{(1)}_{r_p}|}}\Big|W_{\Pi^{(r_p)}_{|\mathcal{G}^{(1)}_{r_p}|}},  \mathcal{H}\right)}_{\mathcal{\tilde{V}}_1(\mathcal{S}_1)}.
\end{align}

Starting with the last three terms in \eqref{RateSum2}, we follow the same logic as in the proof of Theorem \ref{Theorem4UpperBound1}.  Clearly, the aim here should be to choose the indexes for signal components corresponding to the set choices $\mathcal{S}_1$ and  $\mathcal{S}_2$ in a way that allows the maximum cancellation of the entropy terms, while taking advantage of the signal distribution properties stated in Remark \ref{remark33}. Depending upon the network topology $\mathcal{G}$, this cancellation will then result in zero or more positive entropy terms. We calculate the number of these remaining positive entropy terms as shown next. 

To calculate the number of (positive) entropy terms resulting from $\sum_{m_q\in\mathcal{S}_2}\sum_{i\in\mathcal{T}_{m_q}}h(\tilde{Y}^n_{i}|W_i,  \mathcal{H})$, we can use the cardinality of the collection of all index sets of transmitters whose signals are respectively seen at the receivers in $\mathcal{S}_2=\{m_1, m_2, \dots, m_{|\mathcal{S}_2|}\}$. We denote this multi-set of transmitter indexes as $\mathcal{U}(\mathcal{S}_2)=\{\mathcal{T}_{m_q}\}_{q=1}^{|\mathcal{S}_2|}$. To calculate the number of (negative) entropy terms resulting from $\sum_{\Pi^{(m_q)}_{|\mathcal{G}^{(2)}_{m_q}|}\in\mathcal{I}_{m_q}:m_q\in\mathcal{S}_2}h(\tilde{Y}^n_{\Pi^{(m_q)}_{|\mathcal{G}^{(2)}_{m_q}|}}|W_{\Pi^{(m_q)}_{|\mathcal{G}^{(2)}_{m_q}|}},  \mathcal{H})$, we use the cardinality of the set denoted by $\mathcal{\tilde{V}}_2(\mathcal{S}_2)= \{\Pi^{(m_1)}_{|\mathcal{G}^{(2)}_{m_1}|}, \Pi^{(m_2)}_{|\mathcal{G}^{(2)}_{m_2}|}, \dots \Pi^{(m_{|\mathcal{S}_2|})}_{|\mathcal{G}^{(2)}_{m_{|\mathcal{S}_2|}}|}\}$, which is an arbitrary set of indexes $ \Pi^{(m_q)}_{|\mathcal{G}^{(2)}_{m_q}|}\in \mathcal{I}_{m_q}$ for $m_q\in \mathcal{S}_2$. Similarly, to calculate the number of (negative) entropy terms resulting from $\sum_{\Pi^{(r_p)}_{|\mathcal{G}^{(1)}_{r_p}|}\in\mathcal{I}_{r_p}:r_p\in\mathcal{S}_1} h(\tilde{Y}^n_{\Pi^{(r_p)}_{|\mathcal{G}^{(1)}_{r_p}|}}|W_{\Pi^{(r_p)}_{|\mathcal{G}^{(1)}_{r_p}|}},  \mathcal{H})$, we use the cardinality of the set denoted by $\mathcal{\tilde{V}}_1(\mathcal{S}_1))= \{\Pi^{(r_1)}_{|\mathcal{G}^{(1)}_{r_1}|}, \Pi^{(r_2)}_{|\mathcal{G}^{(1)}_{r_2}|}, \dots \Pi^{(r_{|\mathcal{S}_1|})}_{|\mathcal{G}^{(1)}_{r_{|\mathcal{S}_1|}}|}\}$, which is an arbitrary set of indexes $ \Pi^{(r_p)}_{|\mathcal{G}^{(1)}_{r_p}|}\in \mathcal{I}_{r_p}$ for $r_p\in \mathcal{S}_1$. From the defined above three sets, we can then calculate the resulting number of (positive) entropy terms as follows:
\begin{align}
\label{EntropySetCounter2}
\tilde{\mu} (\mathcal{S}_1, \mathcal{S}_2) =|(\mathcal{U}(\mathcal{S}_2) \setminus  \mathcal{\tilde{V}}_2(\mathcal{S}_2)) \setminus \mathcal{\tilde{V}}_1(\mathcal{S}_1))|.
\end{align}

As indicated in \eqref{PowerConstraint},  each signal is transmitted with an average power constraint $P$. Thus, each of the $\tilde{\mu} (\mathcal{S}_1, \mathcal{S}_2)$ entropy terms can be upper bounded by $n\log(P)$. Therefore, using \eqref{RateSum2} and \eqref{EntropySetCounter2}, we can upper bound the secrecy rate as follows:
\begin{align}
\label{BlockSymmetricRate2}
     nR^{\textsf{sym}}&\leq \frac{|\mathcal{S}_1|n\log(P)+\tilde{\mu} (\mathcal{S}_1, \mathcal{S}_2)n\log(P)}{|\mathcal{S}_1|+\sum_{p=1}^{|\mathcal{S}_1|}|\mathcal{G}^{(1)}_{r_p}|+\sum_{q=1}^{|\mathcal{S}_2|}|\mathcal{G}^{(2)}_{m_q}|}\nonumber\\& \hspace{15pt}+\frac{no(\log(P))+no(n)}{|\mathcal{S}_1|+\sum_{p=1}^{|\mathcal{S}_1|}|\mathcal{G}^{(1)}_{r_p}|+\sum_{q=1}^{|\mathcal{S}_2|}|\mathcal{G}^{(2)}_{m_q}|}.
\end{align}
Then, by dividing both sides of  \eqref{BlockSymmetricRate2} by $n\log(P)$, minimizing over all choices of $(\mathcal{S}_1, \mathcal{S}_2, \mathcal{G}^{\mathcal{S}_1}, \mathcal{G}^{\mathcal{S}_2})$, and taking the limit as $n\rightarrow \infty$ and $P\rightarrow \infty$, we obtain the following upper bound on symmetric SDoF:
\begin{align}
\label{SDoFProofExpressionThm5} 
&\textup{\textsf{SDoF}}^{\textup{\textsf{sym}}} \nonumber\\&\leq \min_{(\mathcal{S}1, \mathcal{S}_2)} \min_{(\mathcal{G}^{\mathcal{S}_1}, \mathcal{G}^{\mathcal{S}_2})}\frac{|\mathcal{S}_1|+\tilde{\mu} (\mathcal{S}_1, \mathcal{S}_2)}{|\mathcal{S}_1|+\sum_{p=1}^{|\mathcal{S}_1|}|\mathcal{G}^{(1)}_{r_p}|+\sum_{q=1}^{|\mathcal{S}_2|}|\mathcal{G}^{(2)}_{m_q}|},
\end{align}
which completes the proof of Theorem \ref{Theorem5UpperBound2}.  \QEDA

\section{Proof of Lemma \ref{Lemma1}}
\label{appendix:Lemma1Proof}

Let $\tilde{Y}^n_i=\mathbf{H}^n_{ki}X^n_i+Z^n_k$ , for $i\in\mathcal{T}_k=\{k\}\cup\mathcal{I}_k$. To prove Lemma \ref{Lemma1}, we proceed as follows:
\begin{align}
\label{Lemma1Equation1}
h&(Y^n_k|W_{\mathcal{I}_k}, \mathcal{H})= I(W_k; Y^n_k| W_{\mathcal{I}_k}, \mathcal{H})\nonumber\\&\hspace{70pt}+h(Y^n_k|W_{\{k\}\cup\mathcal{I}_k}, \mathcal{H})\\
\label{Lemma1Equation2}
&\leq I(W_k; Y^n_k, \tilde{Y}^n_{\mathcal{I}_k}| W_{\mathcal{I}_k}, \mathcal{H})\nonumber\\&\hspace{15pt}+h(Y^n_k|W_{\{k\}\cup\mathcal{I}_k}, \mathcal{H})\\
\label{Lemma1Equation3}
&= I(W_k; Y^n_k| W_{\mathcal{I}_k}, \mathcal{H})+I(W_k, \tilde{Y}^n_{\mathcal{I}_k} | Y^n_k, W_{\mathcal{I}_k}, \mathcal{H})\nonumber\\&\hspace{15pt}+h(Y^n_k|W_{\{k\}\cup\mathcal{I}_k}, \mathcal{H})\\
\label{Lemma1Equation31}
&= H(W_k|W_{\mathcal{I}_k}, \mathcal{H})-H(W_k|Y^n_k, W_{\mathcal{I}_k}, \mathcal{H})\nonumber\\&\hspace{15pt}+I(W_k, \tilde{Y}^n_{\mathcal{I}_k} | Y^n_k, W_{\mathcal{I}_k}, \mathcal{H})+h(Y^n_k|W_{\{k\}\cup\mathcal{I}_k}, \mathcal{H})\\
\label{Lemma1Equation4}
&= H(W_k|\mathcal{H})-H(W_k|Y^n_k, \mathcal{H})\nonumber\\&\hspace{15pt}+I(W_k, \tilde{Y}^n_{\mathcal{I}_k} | Y^n_k, W_{\mathcal{I}_k}, \mathcal{H})+h(Y^n_k|W_{\{k\}\cup\mathcal{I}_k}, \mathcal{H})\\
\label{Lemma1Equation5}
&= nR_k+no(n)+I(W_k, \tilde{Y}^n_{\mathcal{I}_k} | Y^n_k, W_{\mathcal{I}_k}, \mathcal{H})\nonumber\\&\hspace{15pt}+h(Y^n_k|W_{\{k\}\cup\mathcal{I}_k}, \mathcal{H})\\
\label{Lemma1Equation6}
&= nR_k+no(n)+h(\tilde{Y}^n_{\mathcal{I}_k} | Y^n_k, W_{\mathcal{I}_k}, \mathcal{H})\nonumber\\&\hspace{15pt}-h(\tilde{Y}^n_{\mathcal{I}_k} | Y^n_k, W_{\{k\}\cup\mathcal{I}_k}, \mathcal{H})+h(Y^n_k|W_{\{k\}\cup\mathcal{I}_k}, \mathcal{H})\\
\label{Lemma1Equation7}
&\leq nR_k+no(n)+h(\tilde{Y}^n_{\mathcal{I}_k} |W_{\mathcal{I}_k}, \mathcal{H})\nonumber\\&\hspace{15pt}-h(\tilde{Y}^n_{\mathcal{I}_k} | Y^n_k, W_{\{k\}\cup\mathcal{I}_k}, \mathcal{H})+h(Y^n_k|W_{\{k\}\cup\mathcal{I}_k}, \mathcal{H})\\
\label{Lemma1Equation71}
&\leq nR_k+no(n)+h(\tilde{Y}^n_{\mathcal{I}_k} |W_{\mathcal{I}_k}, \mathcal{H})\nonumber\\&\hspace{15pt}-h(\tilde{Y}^n_{\mathcal{I}_k} | Y^n_k, W_{\{k\}\cup\mathcal{I}_k}, X^n_{\mathcal{I}_k}, \mathcal{H})\nonumber\\&\hspace{15pt}+h(Y^n_k|W_{\{k\}\cup\mathcal{I}_k}, \mathcal{H})\\
\label{Lemma1Equation8}
&=nR_k+no(n)+h(\tilde{Y}^n_{\mathcal{I}_k} |W_{\mathcal{I}_k}, \mathcal{H})\nonumber\\&\hspace{15pt}-h(\tilde{Y}^n_{\mathcal{I}_k} | Y^n_k, X^n_{\mathcal{I}_k}, W_k, \mathcal{H})+h(Y^n_k|W_{\{k\}\cup\mathcal{I}_k}, \mathcal{H})\\
\label{Lemma1Equation9}
&= nR_k+no(n)+h(\tilde{Y}^n_{\mathcal{I}_k} |W_{\mathcal{I}_k}, \mathcal{H})\nonumber\\&\hspace{15pt}-h(Z^n_{\mathcal{I}_k} | Y^n_k, W_k, \mathcal{H})+h(Y^n_k|W_{\{k\}\cup\mathcal{I}_k}, \mathcal{H})\\
\label{Lemma1Equation10}
&\leq nR_k+no(n)+h(\tilde{Y}^n_{\mathcal{I}_k} |W_{\mathcal{I}_k}, \mathcal{H})\nonumber\\&\hspace{15pt}-h(Z^n_{\mathcal{I}_k})+h(Y^n_k|W_{\{k\}\cup\mathcal{I}_k}, \mathcal{H})\\
\label{Lemma1Equation11}
&= nR_k+no(n)+h(\tilde{Y}^n_{\mathcal{I}_k} |W_{\mathcal{I}_k}, \mathcal{H})\nonumber\\&\hspace{15pt}+no(\log (P))+h(Y^n_k|W_{\{k\}\cup\mathcal{I}_k}, \mathcal{H})\\
\label{Lemma1Equation12}
&\leq nR_k+no(n)+h(\tilde{Y}^n_{\mathcal{I}_k} |W_{\mathcal{I}_k}, \mathcal{H})\nonumber\\&\hspace{15pt}+no(\log (P))+h(Y^n_k|W_k, \mathcal{H})\\
\label{Lemma1Equation13}
&= nR_k+no(n)+\sum_{j\in\mathcal{I}_k}h(\tilde{Y}^n_j |W_j, \mathcal{H})\nonumber\\&\hspace{15pt}+no(\log (P))+h(Y^n_k|W_k, \mathcal{H})\\
\label{Lemma1Equation14}
&= nR_k+\sum_{j\in\mathcal{T}_k}h(\tilde{Y}^n_j |W_j, \mathcal{H})+no(\log (P))+no(n),
\end{align}
where \eqref{Lemma1Equation3} is due to the chain rule of mutual information, \eqref{Lemma1Equation4} follows from the independence of the message $W_k$ from the interfering message set $W_{\mathcal{I}_k}$,  \eqref{Lemma1Equation5} is due to the decodability constraint in \eqref{DecodabilityConstraint}, \eqref{Lemma1Equation7} and \eqref{Lemma1Equation71} follow from the fact that conditioning reduces entropy, \eqref{Lemma1Equation8} is due to the Markov chain $W_{\mathcal{I}_k}\rightarrow X^n_{\mathcal{I}_k}\rightarrow \tilde{Y}^n_{\mathcal{I}_k}$, \eqref{Lemma1Equation10} and \eqref{Lemma1Equation12} follow from the fact that conditioning reduces entropy, \eqref{Lemma1Equation13} is due to the independence of $W_j$ from the message set $W_{\mathcal{I}_k\setminus\{j\}}$, and \eqref{Lemma1Equation14} follows because $\mathcal{T}_k=\{k\}\cup\mathcal{I}_k$.
 \QEDA

\section{Proof of Lemma \ref{Lemma2}}
\label{appendix:Lemma2Proof}

 Let $\tilde{Y}^n_i=\mathbf{H}^n_{ki}X^n_i+Z^n_k$ , for all $i\in\mathcal{T}_k=\{k\}\cup\mathcal{I}_k$, and consider any set $\mathcal{G}_k\subseteq\mathcal{I}_k$ and a permutation sequence  $\Pi^{(k)}=(\Pi_1, \Pi_2, \dots \Pi_{|\mathcal{G}_k|})$ of the elements of $\mathcal{G}_k$ satisfying Definition \ref{FractionalGenerators}. We proceed with the proof of Lemma \ref{Lemma2} as follows:
\begin{align}
\label{Lemma2Equation1}
h&\left(\sum_{i\in\mathcal{I}_k}\mathbf{H}^n_{ki}X^n_i+Z^n_k\Big|\mathcal{H}\right)\nonumber\\&\hspace{0pt}\geq h\left(\sum_{i\in\mathcal{I}_k}\mathbf{H}^n_{ki}X^n_i+Z^n_k\Big|X_{\mathcal{I}_k\setminus\mathcal{G}_k}, \mathcal{H}\right)\\
\label{Lemma2Equation2}
&= h\left(\sum_{\Pi_i\in\mathcal{G}_k}\mathbf{H}^n_{k\Pi_i}X^n_{\Pi_i}+Z^n_k\Big| \mathcal{H}\right)\\
\label{Lemma2Equation3}
&=I\left(W_{\Pi_1}; \sum_{\Pi_i\in\mathcal{G}_k}\mathbf{H}^n_{k\Pi_i}X^n_{\Pi_i}+Z^n_k\Big| \mathcal{H}\right)\nonumber\\&\hspace{15pt}+h\left(\sum_{\Pi_i\in\mathcal{G}_k}\mathbf{H}^n_{k\Pi_i}X^n_{\Pi_i}+Z^n_k\Big|W_{\Pi_1}, \mathcal{H}\right)\\
\label{Lemma2Equation4}
&= H(W_{\Pi_1}| \mathcal{H})\nonumber\\&\hspace{15pt}-H\left(W_{\Pi_1}\Big|\sum_{\Pi_i\in\mathcal{G}_k}\mathbf{H}^n_{k\Pi_i}X^n_{\Pi_i}+Z^n_k,  \mathcal{H}\right)\nonumber\\&\hspace{15pt}+h\left(\sum_{\Pi_i\in\mathcal{G}_k}\mathbf{H}^n_{k\Pi_i}X^n_{\Pi_i}+Z^n_k\Big|W_{\Pi_1},  \mathcal{H}\right)\\
\label{Lemma2Equation5}
&= nR^{\textsf{\textup{sym}}}-H\left(W_{\Pi_1}\Big|\sum_{\Pi_i\in\mathcal{G}_k}\mathbf{H}^n_{k\Pi_i}X^n_{\Pi_i}+Z^n_k,  \mathcal{H}\right)\nonumber\\&\hspace{15pt}+h\left(\sum_{\Pi_i\in\mathcal{G}_k}\mathbf{H}^n_{k\Pi_i}X^n_{\Pi_i}+Z^n_k\Big|W_{\Pi_1}, \mathcal{H}\right)\\
\label{Lemma2Equation6}
&= nR^{\textsf{\textup{sym}}}+no(n)\nonumber\\&\hspace{15pt}+h\left(\sum_{\Pi_i\in\mathcal{G}_k}\mathbf{H}^n_{k\Pi_i}X^n_{\Pi_i}+Z^n_k\Big|W_{\Pi_1},  \mathcal{H}\right)\\
\label{Lemma2Equation61}
&\geq nR^{\textsf{\textup{sym}}}+no(n)\nonumber\\&\hspace{15pt}+h\left(\sum_{\Pi_i\in\mathcal{G}_k}\mathbf{H}^n_{k\Pi_i}X^n_{\Pi_i}+Z^n_k\Big|W_{\Pi_1}, X_{\Pi_1}, \mathcal{H}\right)\\
\label{Lemma2Equation7}
&= nR^{\textsf{\textup{sym}}}+no(n)\nonumber\\&\hspace{15pt}+h\left(\sum_{\Pi_i\in\mathcal{G}_k}\mathbf{H}^n_{k\Pi_i}X^n_{\Pi_i}+Z^n_k\Big|X_{\Pi_1}, \mathcal{H}\right)\\
\label{Lemma2Equation8}
&=  nR^{\textsf{\textup{sym}}}+no(n)\nonumber\\&\hspace{15pt}+h\left(\sum_{\Pi_i\in\mathcal{G}_k\setminus\Pi_1}\mathbf{H}^n_{k\Pi_i}X^n_{\Pi_i}+Z^n_k\Big|\mathcal{H}\right)\\
\label{Lemma2Equation9}
&\geq n |\mathcal{G}_k|R^{\textsf{\textup{sym}}}+h(\tilde{Y}^n_{\Pi_{|\mathcal{G}_k|}}|W_{\Pi_{|\mathcal{G}_k|}},  \mathcal{H})+no(n),
\end{align}
where $\Pi_{|\mathcal{G}_k|}\in\mathcal{G}_k\subseteq \mathcal{I}_k$  \eqref{Lemma2Equation1} follows from the fact that conditioning reduces entropy, \eqref{Lemma2Equation6} is due to Definition \ref{FractionalGenerators}'s fractional signal generator condition $\mathcal{G}_k\setminus\{\Pi_1, \Pi_2, \dots, \Pi_i\}\subseteq \mathcal{I}_{\Pi_i}$, for $i=1, 2, \dots, |\mathcal{G}_k|$, and the decodability constraint in \eqref{DecodabilityConstraint}, \eqref{Lemma2Equation61} is due to the fact that conditioning reduces entropy, \eqref{Lemma2Equation7} is due to the Markov chain $W_{\Pi_1}\rightarrow X^n_{\Pi_1}\rightarrow Y^n_{\Pi_1}$, whereas  \eqref{Lemma2Equation9} is obtained by repeatedly  applying steps \eqref{Lemma2Equation2}-\eqref{Lemma2Equation8} for a total of $|\mathcal{G}_k|$ times.\QEDA
\end{appendices}

\bibliographystyle{IEEEtran}
\bibliography{TIM-CM-Paper-bib}

\end{document}